%% file: rdthomeoj.tex
\documentclass[letterpaper,USenglish,nolineno]{socg-lipics-v2021}

\pdfoutput=1 %uncomment to ensure pdflatex processing (mandatatory e.g. to submit to arXiv)
\hideLIPIcs  %uncomment to remove references to LIPIcs series (logo, DOI, ...), e.g. when preparing a pre-final version to be uploaded to arXiv or another public repository

\usepackage{newtxmath}
\title{Better Sampling Bounds for \\ Restricted Delaunay Triangulations and a \\
       Star-Shaped Property for Restricted Voronoi Cells}

\titlerunning{Better Sampling Bounds for Restricted Delaunay Triangulations}
\author{Jonathan Richard Shewchuk}{University of California, Berkeley, USA}{jrs@berkeley.edu}{https://orcid.org/0009-0004-6864-8994}{}
\authorrunning{J.\,R. Shewchuk}

\Copyright{Jonathan Richard Shewchuk}



\category{} %optional, e.g. invited paper

\relatedversion{} %optional, e.g. full version hosted on arXiv, HAL, or other respository/website
\funding{
Supported by the National Science Foundation under
Awards CCF-0635381, CCF-1423560, and CCF-1909204.
}
\acknowledgements{
I thank Nina Amenta, Jean-Daniel Boissonnat, Siu-Wing Cheng, Tamal Dey,
Arijit Ghosh, and Marc Khoury for discussions about surface sampling.
This work was set in motion by a sabbatical at
INRIA Sophia-Antipolis during the spring of 2010.
I thank the Geometrica Group for their kind reception.}

\newcommand{\Vor}{\mathrm{Vor}}
\newcommand{\Del}{\mathrm{Del}}
\newcommand{\Int}{\mathrm{int}\,}
\newcommand{\Bd}{\mathrm{bd}\,}

\newcommand{\aff}{\mathrm{aff}\,}
\newcommand{\lfs}{\mathrm{lfs}}

\newcommand{\R}{\mathbb{R}}

\begin{document}

\maketitle

\begin{abstract}
% We introduce a new paradigm that will revolutionize
% the way pardigms are revolutionized.
% Synergy will take place too.
The restricted Delaunay triangulation of a closed surface $\Sigma$ and
a finite point set $V \subset \Sigma$ is a subcomplex of
the Delaunay tetrahedralization of $V$ whose triangles approximate $\Sigma$.
It is well known that if $V$ is
a sufficiently dense sample of a smooth $\Sigma$, then
the union of the restricted Delaunay triangles is homeomorphic to $\Sigma$.
We show that an $\epsilon$-sample with $\epsilon \leq 0.3245$ suffices.
By comparison, Dey proves it for a $0.18$-sample;
%%%JRS:  The 0.18-sample fact is on page 70 of Tamal's book (PDF), or
%%%      page 57 of the paper book.
% , which is a more stringent requirement;
our improved sampling bound reduces the number of sample points required by
a factor of $3.25$.
More importantly,
we improve a related sampling bound of Cheng et al.\ for
Delaunay surface meshing, reducing the number of sample points required
by a factor of~$21$.
The first step of our homeomorphism proof is particularly interesting:
we show that for a $0.44$-sample,
the restricted Voronoi cell of each site $v \in V$ is homeomorphic to a disk,
and the orthogonal projection of the cell onto
$T_v\Sigma$ (the plane tangent to $\Sigma$ at $v$) is star-shaped.
% That result generalizes to
% manifolds of higher dimension embedded in higher-dimensional spaces.
\end{abstract}

\section{Introduction}
\label{intro}

The {\em restricted Delaunay triangulation} (RDT) is
a well-established way of generating good-quality triangulations
on curved surfaces~\cite{edelsbrunner97}.
Researchers have developed a theory of surface sampling to determine
how we should sample points on a surface to guarantee that
an RDT (or a related triangulation) is
a topologically correct and geometrically accurate approximation of the surface
\cite{amenta99b,amenta98b,amenta02,boissonnat18,boissonnat14,boissonnat03,
      boissonnat05,boltcheva17,cheng01,cheng07,cheng05,cheng12,dey07,
      edelsbrunner97,khoury16,khoury21,pellerin14,wang22,yan09}.
RDTs and this surface sampling theory have equipped geometers
to rigorously prove the correctness of algorithms for
surface reconstruction~\cite{dey07} and surface mesh generation~\cite{cheng12}.

Think of the RDT as a function that takes in two inputs:
a smooth, closed (compact with no boundary) surface
$\Sigma \subset \mathbb{R}^3$ and
a finite set $V \subset \Sigma$ of points, called
\emph{sites} (or \emph{vertices}~of the RDT).
$V$ is called a {\em sample} or a {\em point cloud}.
The output is a simplicial complex ${\mathcal T}$ whose vertices are $V$.
The RDT ${\mathcal T}$ is
a subcomplex of the three-dimensional Delaunay triangulation $\Del \, V$, but
in typical usage ${\mathcal T}$~contains no tetrahedra;
only triangles, edges, and the vertices~$V$.

If $V$ is sufficiently dense,
${\mathcal T}$ is a (topological) {\em triangulation} of $\Sigma$,
which means that the \emph{underlying space} of $\mathcal{T}$, written
$|\mathcal{T}| = \bigcup_{\tau \in \mathcal{T}} \tau$, is homeomorphic to $\Sigma$.
This paper is about proving that
a modest sampling requirement suffices to guarantee that homeomorphism.

What does it mean for $V$ to be ``sufficiently dense''?
%%%CUT (one phrase)
Intuitively,
% for approximating the shape and topology of a surface from a point cloud,
there should be no large unsampled ``bare spots'' on $\Sigma$.
Ideally, sampling requirements are adaptive, as
a denser spacing of sites is needed in
regions where $\Sigma$ has higher curvature or
closely-spaced ``parts'' like fingers on a hand, but
simple regions need few sites.
Many provably good algorithms for surface reconstruction assume that
the sites $V$ are a so-called $\epsilon$-sample of $\Sigma$
(defined in Section~\ref{rdts}), an adaptive sample in which
a smaller $\epsilon$ implies more sites, closer together.

One main result of this paper is that if $V$ is a $0.3245$-sample of $\Sigma$,
the underlying space of the RDT is homeomorphic to $\Sigma$.
Dey~\cite{dey07} proved the same for a $0.18$-sample.
% the best previous bound we are aware of.
The new result reduces the number of sample points required by
a factor of $3.25$ (the square of $0.3245 / 0.18$).
Some well-known surface reconstruction algorithms such as
the Crust~\cite{amenta99b} and Cocone~\cite{amenta02} algorithms
rely on identifying a superset of the RDT's triangles then paring them down.
Any substantial relaxation of their sampling requirements is good news for
a broad swath of existing algorithms, and
it helps to explain why they work well in practice.

As a point of comparison, Bjerkevik~\cite{bjerkevik22} shows that
no proof will ever guarantee homeomorphism for $0.72$-samples, as
there exist $0.72$-samples with
multiple, topologically different, correct reconstructions.
(This limitation holds even for smooth, closed curves in the plane.)

We define another adaptive sampling condition better suited to
mesh generation,
%%%UNCUT
% algorithms,
enabling stronger sampling bounds:
we prove homeomorphism for what we call a {\em $0.4132$-Voronoi sample}.
Cheng et al.~\cite{cheng12} proved the same for a $0.09$-Voronoi sample;
our result reduces the number of sample points required by
a factor of~$21$.
Our bound implies better sampling bounds for
existing surface meshing algorithms; it is this paper's most important result.

Also of interest are the new techniques introduced here to obtain these bounds.
In particular, for a $0.44$-sample or a $0.78$-Voronoi sample,
the restricted Voronoi cell of each site $v \in V$
(defined in Section~\ref{rdts}) is homeomorphic to a disk.
Moreover, let $T_v\Sigma$ denote the plane tangent to $\Sigma$ at $v$;
the orthogonal projection of $v$'s cell onto $T_v\Sigma$ is {\em star-shaped}:
a union of (infinitely many) line segments terminating at $v$.
% {\em star-shaped}:  for every point $y \in \varphi(\Vor|_\Sigma v)$, the
% line segment connecting $v$ to $y$ is a subset of $\varphi(\Vor|_\Sigma v)$.
It seems a bit surprising that restricted Voronoi cells are better behaved
with respect to coarse samples than restricted Voronoi vertices or edges,
because the proofs by Dey~\cite{dey07} and Cheng et al.~\cite{cheng12} use
the soundness of the restricted Voronoi edges to establish
the soundness of the restricted Voronoi cells.
This paper reverses that sequence and, in my opinion, gets at the heart of
the reasons why a restricted Voronoi cell is nicely shaped.
Remarkably, the results about the Voronoi cells generalize to
manifolds of higher dimension embedded in higher-dimensional spaces
with no degradation in the bounds (see Section~\ref{higherd}),
although the homeomorphism result does not and cannot generalize
to higher dimensions~\cite{cheng05,boissonnat09,boissonnat18b}.

Not everyone gets excited by improved constants, but
I advocate for the importance of
tighter sampling theory for computational geometry, especially for
our most fundamental notions such as RDTs---just as
numerical analysts have devoted much effort to improving
constants associated with quadrature rules, interpolation theory,
differential equation solvers, and more.
%%% JRS UNCUT
% We see our work as interpolation theory,
% but for surfaces rather than functions.
Reasonable constants show that RDTs are useful rather than merely theoretical.

\section{Restricted Delaunay triangulations and ${\bf \epsilon}$-samples}
\label{rdts}

The RDT is defined by dualizing a {\em restricted Voronoi diagram}, which
will be our main object of study.
Let $|pq|$ denote the Euclidean distance from $p$ to $q$;
equivalently, the length of the line segment $pq$.
%Let $pq$ denote the line segment connecting $p$ to $q$ and
%$|pq|$ denote its length; equivalently, the Euclidean distance from $p$ to $q$.
Given a closed surface $\Sigma \subset \R^3$ and a point set $V \subset \Sigma$,
the {\em restricted Voronoi cell} of a site $v \in V$ is
$\Vor|_\Sigma v = \{ p \in \Sigma : |pv| \leq |pw|$ for all $w \in V \}$.
% denoted $\Vor|_\Sigma v, is the set of all points on $\Sigma$ for which
% $v$~is the closest site in $V$ (possibly tied for closest),
% as measured by the Euclidean distance in $\R^3$.
Equivalently, $\Vor|_\Sigma v = \Sigma \cap \Vor \, v$, where
$\Vor \, v$ is $v$'s standard Voronoi cell in $\R^3$.
The name ``restricted Voronoi cell'' means that $\Vor|_\Sigma v$ is
the restriction of $\Vor \, v$ to the surface~$\Sigma$.
See Figure~\ref{vordel}.

\begin{figure}
\centerline{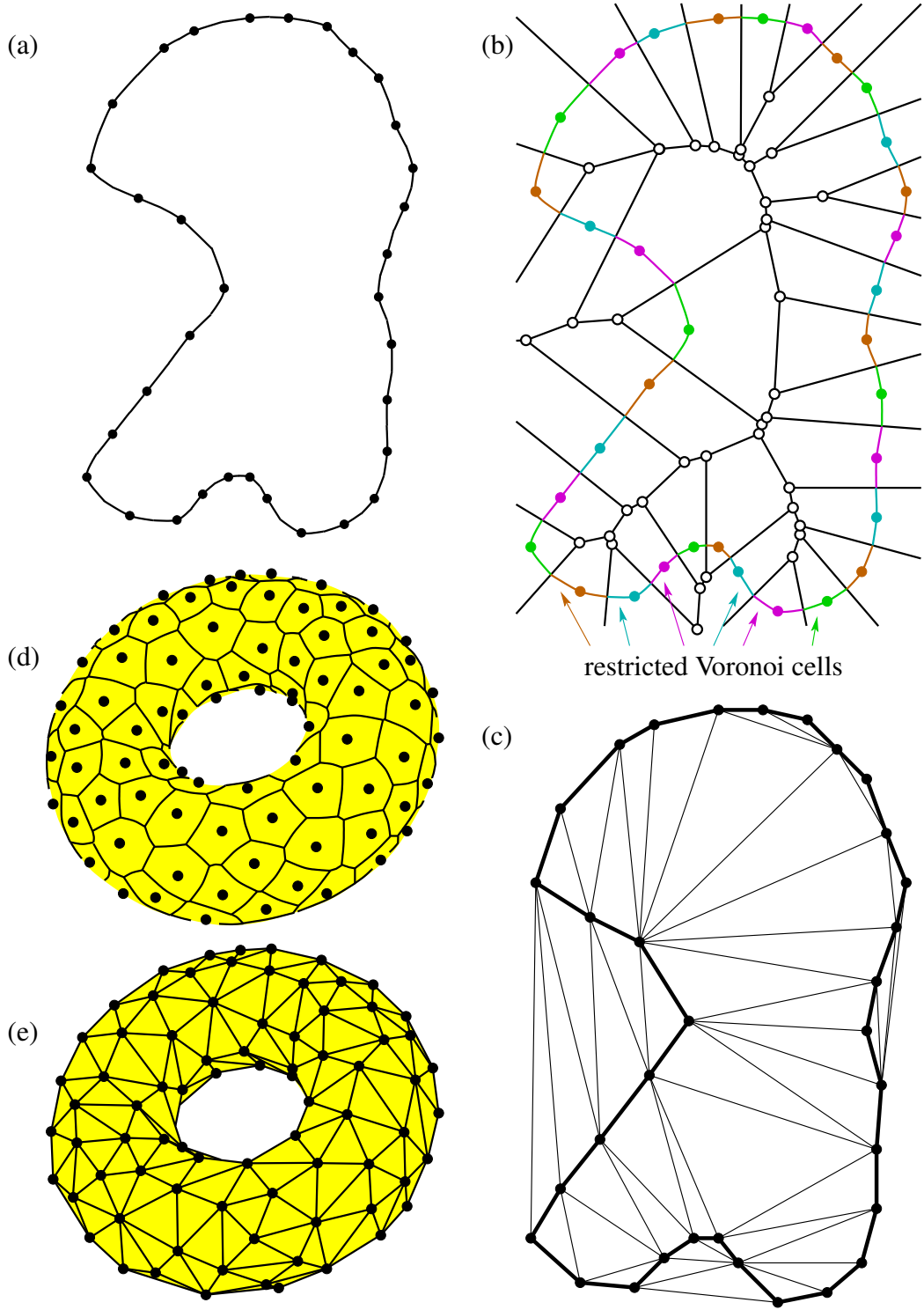}

\caption{\label{vordel}
(a)~A two-dimensional view of restricted Delaunay triangulations.
The input is a smooth, closed curve $\Sigma$ and a sample
$V \subset \Sigma$.
(b)~The restricted Voronoi diagram is the restriction of
the (classic) Voronoi diagram to $\Sigma$.
(c)~The restricted Delaunay triangulation (bold) is
the dual of the restricted Voronoi diagram and
a subcomplex of the (classic) Delaunay triangulation.
(d)~A~restricted Voronoi diagram in three dimensions.
(e)~Its dual restricted Delaunay triangulation.
}
\end{figure}

A {\em restricted Voronoi face} is any nonempty intersection of
one or more restricted Voronoi cells
(i.e., $\Sigma \cap F$ for some face $F$ in $\Vor \, v$.)
In particular,
a {\em restricted Voronoi vertex} is the nonempty intersection of $\Sigma$ with
an edge in $\Vor \, V$; and
a {\em restricted Voronoi edge} is the nonempty intersection of $\Sigma$ with
a polygonal $2$-face in $\Vor \, V$.
Typically these entities are a single point and a (curvy) path on~$\Sigma$,
respectively, but if $V$ is not dense enough,
they may take on more pathological forms---e.g.,
a~restricted Voronoi ``edge'' could be a cycle, a pair of disjoint paths, or
even a $2$-dimensional blob.
This paper aims to determine
sampling conditions that eliminate such pathologies.
% (In some circumstances, it is useful to declare that
% the empty set is a Voronoi face of dimension $-1$.)
The \emph{restricted Voronoi diagram} $\Vor|_\Sigma V$ is
the cell complex containing all the restricted Voronoi faces
(including the restricted Voronoi cells).

The \emph{Delaunay subdivision} $\Del \, V$ is
the polyhedral complex dual to the Voronoi diagram, and
the \emph{restricted Delaunay subdivision} $\Del|_\Sigma V$ is
the subcomplex of $\Del \, V$ dual to the restricted Voronoi diagram.
% That is, for each restricted Voronoi face $f \in \Vor|_\Sigma V$,
% let $W \subseteq V$ be the set of sites whose restricted Voronoi cells include
% $f$ and let $f^*$ be the convex hull of $W$;
That is, for each Voronoi face $F \in \Vor \, V$,
let $W \subseteq V$ be the set of sites whose Voronoi cells include
$F$ and let $F^*$ be the convex hull of $W$;
we say that $F^*$ is the face {\em dual} to $F$.
Then $\Del \, V = \{ F^* : F \in \Vor \, V \}$.
The restricted Delaunay subdivision contains the dual faces whose primal faces
intersect $\Sigma$; that is,
$\Del|_\Sigma V = \{ F^* \in \Del \, V : F \cap \Sigma \neq \emptyset \}$.
The restricted Voronoi face $f = F \cap \Sigma$ has
the same dual face as $F$, $f^* = F^*$.
It is customary to subdivide the polyhedra in $\Del \, V$ into tetrahedra
(in which case the duality is no longer strict);
accordingly, we can subdivide the polygons in $\Del|_\Sigma V$ into triangles
and call it a {\em restricted Delaunay triangulation}.
The results in this paper apply whether we do or don't.

A crucial observation in the theory of surface sampling is that
the sampling density necessary for accurate approximation is proportional to
a field called the {\em local feature size}.
%%%UNCUT
% defined in terms of the distance from the surface to its medial axis.
The \emph{medial axis} $M$ of~$\Sigma$, illustrated in Figure~\ref{medial}, is
the closure of the set of all points in $\R^3$ for which
the closest point on $\Sigma$ is not unique.
% Intuitively, the medial axis of $\Sigma$ is meant to capture
% the ``middle'' of the region bounded by $\Sigma$.
A {\em medial ball} is a ball whose center lies on $M$ and
whose boundary intersects $\Sigma$ (tangentially), but
the interior of the ball does not.
For any point $x \in \Sigma$, there are one or two medial balls that have
$x$ on their boundaries, called \emph{medial balls at $x$}.
One is inside $\Sigma$.
If there are two, the other is outside.
%%%CUT
If not, there is an open halfspace tangent to $\Sigma$ at $x$,
disjoint from $\Sigma$, that we call a degenerate ``medial ball.''

\begin{figure}
\centerline{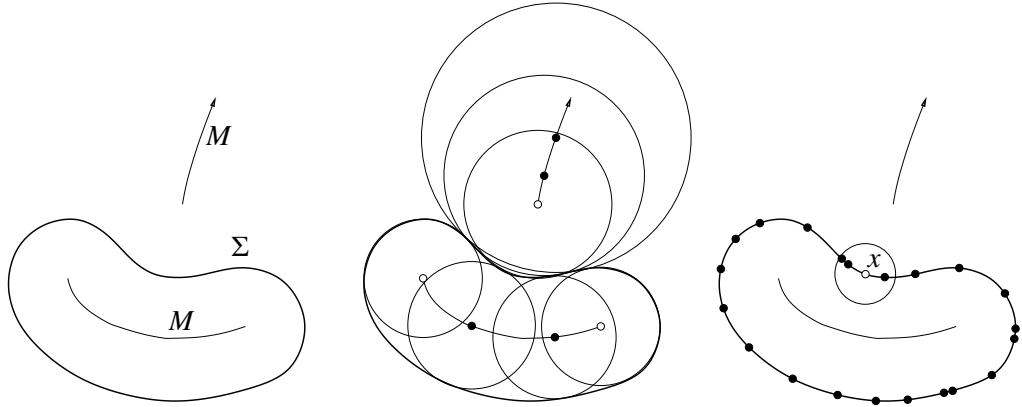}

\caption{\label{medial}  \protect\small \sf
Left:  A $1$-manifold $\Sigma$ and its medial axis~$M$.
Center:  Some of the medial balls that define~$M$.
Balls with black centers touch two points on $\Sigma$.
The white points are in the closure of the black centers.
Right:  A $0.5$-sample of $\Sigma$ (black points).
The ball with center $x$ and radius $0.5 \, \lfs(x)$ contains a site.
}
\end{figure}

The \emph{local feature size} function is
$\lfs: \Sigma \rightarrow \R$, $x \mapsto d(x, M)$ where
% $d(x, M)$ denotes the Euclidean distance from $x$ to $M$.
$d(x, M) = \min_{m \in M} |xm|$.
We require that $\Sigma$ is smooth in the sense that
$\inf_{x \in \Sigma} \lfs(x) > 0$.
($C^{1,1}$-continuity suffices.)

A finite point set $V \subset \Sigma$ is an {\em $\epsilon$-sample of $\Sigma$}
if for every point $x \in \Sigma$,
there is a site $v \in V$ such that $|xv| \leq \epsilon \, \lfs(x)$.
That is, the ball with center $x$ and radius $\epsilon \, \lfs(x)$ contains
at least one site.
% See Figure~\ref{fig:sample}.
A finite $V \subset \Sigma$ is
an {\em $\epsilon$-Voronoi sample of $\Sigma$} if $V \neq \emptyset$ and
for every site $v \in V$ and every point $x \in \Vor|_\Sigma v$,
$|xv| \leq \epsilon \, \lfs(v)$.
That is, $\Vor|_\Sigma v$ is a subset of the ball with center~$v$ and
radius $\epsilon \, \lfs(v)$.
(Note that this condition implies the {\em $\epsilon$-small condition} of
Cheng et al.~\cite{cheng12}, though the converse is not true.)
Only $\epsilon$-samples are a well-known concept, but
$\epsilon$-Voronoi samples are nice because
$\lfs(v)$ tends to give better sampling bounds than $\lfs(x)$.

Like the classical proofs~\cite{amenta99b,cheng12,dey07},
this paper's homeomorphism proofs rely on
the Topological Ball Theorem of Edelsbrunner and Shah~\cite{edelsbrunner97}.
% Proofs about the homeomorphism of restricted Delaunay triangulations
% usually rely on the {\em Topological Ball Theorem} of
% Edelsbrunner and Shah~\cite{edelsbrunner97}, and ours is no exception.
Given a sample $V \subset \Sigma$ of a closed surface $\Sigma \subset \R^3$,
the Topological Ball Theorem states that
$|\Del|_\Sigma V|$ is homeomorphic to $\Sigma$ if
$\Sigma$~and $V$ satisfy two properties.
The {\em closed ball property} is that for each Voronoi face $F \in \Vor \, V$,
$f = F \cap \Sigma$ is either empty or a topological closed $(k - 1)$-ball
% (i.e., homeomorphic to $\mathbb{B}^{k - 1}$)
where $k$ is the dimension of~$F$.
That is,
(A)~for a Voronoi $3$-cell $F$, $f$ is a topological closed disk;
(B)~for a Voronoi polygon~$F$, $f$ is a topological interval or $\emptyset$;
(C)~for a Voronoi edge $F$, $f$ contains at most one point; and
(D)~for a Voronoi vertex $F$, $f = \emptyset$.

If A--D hold,
the {\em generic intersection property} is that for each $F \in \Vor \, V$,
$\Int \, f \subset \Int \, F$ and $\Bd \, f \subset \Bd \, F$.
In this (obscure) definition, ``interior'' and ``boundary'' are interpreted
by the rules of $(k - 1)$-manifolds (for~$f$) and $k$-manifolds (for $F$).
(They are {\em not} the interior and boundary with respect to $\R^3$.)
%%%UNCUT
% Edelsbrunner and Shah~\cite{edelsbrunner97} write this property as
% $\Sigma \cap \Int F = \Int (\Sigma \cap F)$, which is equivalent.)
% \footnote{
% , but I~was long puzzled by what it implied.
% Assume that the closed ball property holds before trying to make sense of
% the generic intersection property; otherwise, you might be misled by
% mistaken notions of what the ``interior'' of $F$ is, as I was for many years.
% }
For a smooth $\Sigma$, properties~A--D and
the following two properties together imply the generic intersection property:
(E)~there is no point where $\Sigma$ intersects
    a polygon of $\Vor \, V$ tangentially
    (i.e., where $F \subset T_p\Sigma$); and
(F)~there is no point where $\Sigma$ intersects
    an edge of $\Vor \, V$ tangentially
    (again, where $F \subset T_p\Sigma$, but $F$ is an edge).
% (G)~$\Sigma$ does not intersect any Voronoi vertex.
% Observe that D and G are the same condition.

This paper is organized around proving that these conditions hold
for a sufficiently dense sample $V$:
properties~A and~E in Section~\ref{2ball},
properties~C and~F in Section~\ref{0ball}, and
property~B in Section~\ref{1ball}.
% properties~A and~E in Section~\ref{2ball} and
% properties~B, C, and~F in Section~\ref{1ball}.
Property~D ensures that $\Del|_\Sigma V$ contains no polyhedra---only
faces of dimension two or less.
Property~D cannot be enforced by dense sampling, but
it can be enforced by an infinitesimal perturbation of $\Sigma$ so
that $\Sigma$ intersects no Voronoi vertex.
This perturbation is easy to simulate symbolically in software simply
by treating each Voronoi vertex on $\Sigma$ as if
it were strictly inside $\Sigma$
(when deciding which faces of $\Del \, V$ are in $\Del|_\Sigma V$).
% (and treating the faces connected to that vertex accordingly).

Unlike in Dey~\cite{dey07} or Cheng et al.~\cite{cheng12},
our proof of property~A does not rely on property~B or~C.
Property~A holds for a $0.4401$-sample or a $0.7861$-Voronoi sample, whereas
we prove property~B only for a $0.3245$-sample or a $0.4132$-Voronoi sample.
Property~B (restricted Voronoi edges) is
the bottleneck that determines our sampling requirements.
% (because
% the directions where Voronoi edges point are not very stable in samples).
% and property~C only for a $0.33$-sample
% (which can be improved to a $0.35$-sample; see Appendix~\ref{better0ball})
% or a $0.49$-Voronoi sample.
% We also have stronger proofs of properties~B and~C.
% each ``restricted Voronoi vertex'' is in fact a lone point and
% each ``restricted Voronoi edge'' is in fact a topological interval ($1$-ball).

Some algorithms for surfaces guarantee topological correctness without RDTs,
by other methods \cite{attali15,boissonnat21,cohensteiner20,kim20,niyogi08},
but in this paper we wish to see how far we can push RDTs.
One alternative is to consider Voronoi diagrams with other distance metrics.
Dyer, Zhang, and M\"{o}ller~\cite{dyer07,dyer08} and
% Boissonnat, Dyer, and Ghosh~\cite{boissonnat18c} use
others~\cite{leibon99,boissonnat18c} use
intrinsic distances within $\Sigma$,
also known as geodesic distances when $\Sigma$ is smooth,
to define {\em intrinsic Voronoi diagrams} that dualize to
{\em intrinsic Delaunay triangulations} (IDTs).
Advantages are that the Voronoi cells are trivially star-shaped, and
the sampling requirements needed to guarantee a homeomorphic triangulation are
mild~\cite{dyer08}.
But intrinsic distances on smooth surfaces are
painful to compute~\cite{patrikalakis02};
RDTs will probably remain popular as an easy alternative.

\section{A relationship between surface points and nearby tangent planes}
\label{tangentsites}

For a smooth, closed surface $\Sigma$ and a point $x \in \Sigma$,
let $T_x\Sigma \subset \R^3$ denote the plane tangent to~$\Sigma$ at~$x$.
% , and let $n_x$ denote the outside-facing vector normal to $\Sigma$ and
% to $T_x\Sigma$ at $x$.
Lemma~\ref{lem:abovebelow}, below, establishes
a relationship between $T_x\Sigma$ and $v$
for two nearby points $v, x \in \Sigma$.
% ; and a relationship between neighborhoods of $x$ and $T_v\Sigma$.
This relationship prepares us to prove in Section~\ref{2ball} that
under suitable sampling conditions,
a restricted Voronoi cell is a topological disk with
a star-shaped projection onto its site's tangent plane.
% The following lemma will give us
% the power to prove that theorem for relatively coarse point samples.
% (To~understand why, it is helpful to know that if $v$ is a site,
% the line segment $oo'$ discussed in the proof is in
% $v$'s three-dimensional Voronoi cell $\Vor \, v$;
% see also Lemma~\ref{lem:raybisector}.)
Lemma~\ref{lem:abovebelow} is surprisingly strong;
the constant $\xi$ (as it is applied in Theorem~\ref{thm:homeocell})
will likely be hard to improve.

Let $B$ and $B'$ be the two open balls of radius $\lfs(v)$ tangent to $\Sigma$
at $v$, and let $o$ and $o'$ be their centers.
As $B$ and $B'$ are subsets of the medial balls at $v$,
they are disjoint from $\Sigma$.

\begin{lemma}
\label{lem:abovebelow}
Consider two points $v, x \in \Sigma$ such that $|vx| < \xi \, \lfs(v)$,
where $\xi = \sqrt{(\sqrt{5} - 1) / 2} \doteq 0.786151$.
% Let $T_x\Sigma$ be the plane tangent to $\Sigma$ at~$x$.
Then $o$ and $o'$ lie on strictly opposite sides of $T_x\Sigma$.
\end{lemma}

\begin{proof}
% If $x = v$ the result follows immediately, so assume that $x \neq v$.
Suppose for the sake of contradiction that
$o$ and $o'$ do not lie on strictly opposite sides of $T_x\Sigma$,
as illustrated in Figure~\ref{oopposite}, left.
%%%CUT
Then $T_x\Sigma \neq T_v\Sigma$ and $x \neq v$.
% , as $T_v\Sigma \cap oo' = \{ v \}$).
Moreover, either $oo' \subset T_x\Sigma$ or
$T_x\Sigma$ does not intersect the relative interior of $oo'$.
% (where $oo'$ is the line segment with endpoints $o$ and $o'$).

Let $B_m$ be the open medial ball tangent to $\Sigma$ at $x$ that is
on the same side of $T_x\Sigma$ as~$v$ (either side if $v \in T_x\Sigma$),
as illustrated.
If $B_m$ is an open halfspace then $T_x\Sigma = \mathrm{boundary} \, B_m$,
$v \in T_x\Sigma$ (as $v \not\in B_m$),
$T_v\Sigma = T_x\Sigma$ (as $\Sigma \cap B_m = \emptyset$), and
the lemma follows.
So assume $B_m$ is bounded.
Its center $m$ lies on $\Sigma$'s medial axis.
The line segment $xm$ is perpendicular to $T_x\Sigma$.

\begin{figure}
\centerline{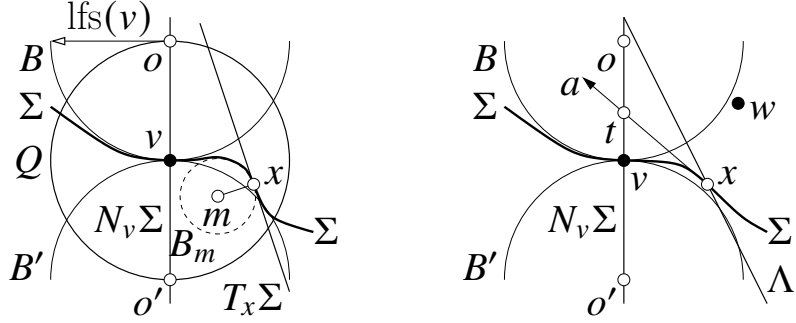}

\caption{\label{oopposite}
Left:  for two sufficiently close points $v, x \in \Sigma$,
suppose that the tangent plane $T_x\Sigma$ does not intersect $oo'$, as shown.
This leads to a contradiction; hence $T_x\Sigma$ must intersect $oo'$.
The medial ball $B_m$ is tangent to $\Sigma$ at $x$.
The center $m$ of $B_m$ cannot lie in the open ball $Q$.
The points $v$, $x$, $o$, and $o'$ lie on the plane of the page, but
$m$ and $w$ generally do not; imagine $m$ floating above the page.
The dashed circle shows the page's cross section of $B_m$, but $B_m$ is larger.
The surface $\Sigma$ cannot intersect the open balls $B$, $B'$, and $B_m$, so
$B$ and $B_m$ are disjoint.
Right:  the plane $\Lambda$ bisects $vw$.
The ray $a \in T_x\Sigma$ intersects the relative interior of $oo'$ and
is strictly on the same side of $\Lambda$ as $v$.
}
\end{figure}

Observe that $v$ is in the relative interior of $oo'$.
If $oo' \subset T_x\Sigma$, then
$\angle vxm = \angle oxm = \angle o'xm = 90^\circ$.
% , because $xm$ is perpendicular to $T_x\Sigma$.
Otherwise, $T_x\Sigma$ does not intersect the relative interior of $oo'$,
% which contains $v$,
so $\angle vxm < 90^\circ$, $\angle oxm \leq 90^\circ$, and
$\angle o'xm \leq 90^\circ$
(the last two because $o$ and $o'$ cannot be
on the side of $T_x\Sigma$ opposite from $v$).
By Pythagoras' Theorem, $|ox|^2 + |mx|^2 \geq |om|^2$ and
$|vx|^2 + |mx|^2 \geq |vm|^2$.
% Let $Q$ be the open ball with center $v$ and radius $\lfs(v)$;
% by the definition of $\lfs(\cdot)$, the medial axis is disjoint from $Q$.
% As $m$ lies on the medial axis, $m \not\in Q$,
% hence $|vm| \geq \lfs(v)$ and thus $|vx|^2 + |mx|^2 > \lfs(v)^2$.
As $m$~lies on the medial axis, $|vm| \geq \lfs(v)$
(by the definition of $\lfs$), so $|vx|^2 + |mx|^2 \geq \lfs(v)^2$.

The surface $\Sigma$ intersects none of the open balls $B$, $B'$, or $B_m$, but
it passes between $B$ and $B'$ at $v$.
As $\Sigma$ is closed and cuts space into two pieces,
one containing $B$ and one containing~$B'$, the ball $B_m$ must lie in
one of those two pieces.
% Suppose without loss of generality that $B_m$ lies in the same piece as $B'$,
Choose the labels $B$ and $B'$ so that $B_m$ lies in the same piece as $B'$,
as illustrated; then $B_m$ must be disjoint from $B$.
% The point $x$ lies either on the upper surface of the medial ball $B_m$ or on
% its lower surface; assume without loss of generality that
% $x$ is on its upper surface, so that $\Sigma \cap F$
% passes above $B_m$ and $B'$ and below $B$, as illustrated.
The radii of $B$ and $B_m$ are $\lfs(v)$ and $|mx|$ respectively, so
$|om| \geq \lfs(v) + |mx|$.
Combining this with the inequality $|ox|^2 + |mx|^2 \geq |om|^2$ gives
% $|ox|^2 + |mx|^2 \geq $(|mx| + \lfs(v))^2$
$|ox|^2 \geq \lfs(v)^2 + 2 \, \lfs(v) \, |mx|$.
Combining this with the inequality $|mx|^2 \geq \lfs(v)^2 - |vx|^2$ gives
$|ox|^2 \geq \lfs(v)^2 + 2 \, \lfs(v) \sqrt{\lfs(v)^2 - |vx|^2}$.

Let $N_v\Sigma$ be the line through $o'$, $v$, and $o$
(the vertical axis in Figure~\ref{oopposite}).
Create a coordinate system with $v = (0, 0, 0)$ and $x = (x_h, x_v, 0)$
such that $x_v$ is the coordinate in the direction parallel to $N_v\Sigma$ and
$x_h$ is the distance from $x$ to $N_v\Sigma$
(the horizontal axis in Figure~\ref{oopposite}).
Then $|ox|^2 + |o'x|^2 =
x_h^2 + (x_v - \lfs(v))^2 + x_h^2 + (x_v + \lfs(v))^2 =
2 x_h^2 + 2 x_v^2 + 2 \, \lfs(v)^2 = 2 \, |vx|^2 + 2 \, \lfs(v)^2$.
Rewrite this as $|vx|^2 = (|ox|^2 + |o'x|^2 - 2 \, \lfs(v)^2) / 2$.
As $x \not\in B'$, $|o'x|^2 \geq \lfs(v)^2$.
Combining these with the inequality
$|ox|^2 \geq \lfs(v)^2 + 2 \, \lfs(v) \sqrt{\lfs(v)^2 - |vx|^2}$ gives
$|vx|^2 \geq \lfs(v) \sqrt{\lfs(v)^2 - |vx|^2}$.

As $|vx| < \xi \, \lfs(v)$, %%% by assumption,
we have $\xi^2 > |vx|^2 / \lfs(v)^2 \geq \sqrt{1 - |vx|^2 / \lfs(v)^2}
> \sqrt{1 - \xi^2}$, which is equivalent to
$\xi^4 + \xi^2 - 1 > 0$, hence $\xi > \sqrt{(\sqrt{5} - 1) / 2}$.
The result follows by contradiction.
\end{proof}

\section{Restricted Voronoi cells are topological disks}
\label{2ball}

This section investigates sampling conditions that guarantee that
(1)~every restricted Voronoi cell has the topology of a closed disk
(closed ball property A),
(2)~the projection of each restricted Voronoi cell onto
its site's tangent plane is star-shaped, and
(3)~no polygon of $\Vor \, V$ intersects $\Sigma$ tangentially
(generic intersection property E).
% Generic intersection property E is proved along the way.
%%%CUT
Theorem~\ref{thm:homeocell} shows that a $0.78$-Voronoi sample suffices, and
Corollary~\ref{cor:homeocelleps} shows that a $0.44$-sample suffices.

Consider a site $v \in V$,
its Voronoi cell $\Vor \, v$, and
its restricted Voronoi cell $\Vor|_\Sigma v = \Sigma \cap \Vor \, v$.
Let $\varphi$~be the map that
orthogonally projects $\R^3$ onto $T_v\Sigma$.
% the plane tangent to $\Sigma$ at $v$.
Note that $\varphi(v) = v$.
% A point $y$ is {\em in the interior of} $\Vor|_\Sigma v$ if
% % $y$ has an open neighborhood $N \subset \Vor|_\Sigma v$ homeomorphic to
% % a disk;
% there is a topological open disk $N \subset \Vor|_\Sigma v$ with $y \in N$;
% {\em on the boundary of} $\Vor|_\Sigma v$ otherwise.

Define a {\em radial path} to be a closed topological interval
%%%UNCUT
%%% (i.e., a topological $1$-ball)
$\gamma \subset \Vor|_\Sigma v$ such that
\begin{itemize}
\item
one endpoint of $\gamma$ is the site $v$,
\item
the other endpoint---call it $z$---lies on the boundary of $\Vor \, v$,
% (the three-dimensional cell),
\item
every point on $\gamma \setminus \{ z \}$ lies in the interior of $\Vor \, v$,
% no point on $\gamma \setminus \{ z \}$ lies on the boundary of $\Vor \, v$,
and
\item
$\varphi|_\gamma$~is a homeomorphism from $\gamma$
to a line segment on $T_v\Sigma$ with endpoints $v$ and $\varphi(z)$, where
$\varphi|_\gamma$ denotes the restriction of $\varphi$ to the domain $\gamma$.
\end{itemize}

We will see that under suitable sampling conditions,
every point in $\Vor|_\Sigma v$ lies on exactly one radial path,
except $v$ itself (Lemma~\ref{lem:pathincell}).
It follows that we can decompose $\Vor|_\Sigma v$ into
radial paths such that no two share a point besides $v$
(Lemma~\ref{lem:pathsincell}).
That is, if we remove~$v$ from each radial path, then
we have a partition of $\Vor|_\Sigma v \setminus \{ v \}$ into paths.
Therefore, $\varphi(\Vor|_\Sigma v)$ is star-shaped.
Although $\Vor|_\Sigma v$ itself is not star-shaped,
its decomposition into radial paths is
a curvy variant of ``star-shaped.''
As the lengths of the projected radial paths vary continuously with
their polar angles, $\Vor|_\Sigma v$ is homeomorphic to a closed disk
(Theorem~\ref{thm:homeocell}).

Let $N_v\Sigma$ be the line normal to $\Sigma$ at $v$
(orthogonal to $T_v\Sigma$).
Let $B$ and $B'$ be the two open balls of radius $\lfs(v)$ tangent to $\Sigma$
at $v$, and let $o$ and $o'$ be their centers. % , respectively.
Then $o, o' \in N_v\Sigma$.

The following lemma implies that
if you are standing on the boundary of $\Vor|_\Sigma v$ and
you walk toward $v$ on a radial path,
you immediately enter the interior of $\Vor \, v$.
(The proof of Lemma~\ref{lem:pathincell} develops this idea further.)
It also implies generic intersection property E.

\begin{lemma}
\label{lem:raybisector}
Consider two distinct sites $v, w \in V$ and a point
$x \in \Vor|_\Sigma v \cap \Vor|_\Sigma w$.
Suppose that $|vx| < \xi \, \lfs(v)$
where $\xi = \sqrt{(\sqrt{5} - 1) / 2} \doteq 0.786151$.
Let $\Lambda$ be the plane that orthogonally bisects the line segment $vw$
(thus $x \in \Lambda$).
By Lemma~\ref{lem:abovebelow}, $T_x\Sigma$ intersects
the relative interior of the line segment~$oo'$ at a lone point~$t$.
Let $a$ be the open ray $\vec{xt}$, and observe that $a \subset T_x\Sigma$.
% Hence $a \subset T_x\Sigma$, $a$ is tangent to $\Sigma$ at $x$,
% and $a$ passes through~$N_v\Sigma$ (i.e., $v \in \varphi(a)$).

Then $v$ and $a$ are strictly on the same side of~$\Lambda$.
\end{lemma}

\begin{proof}
See Figure~\ref{oopposite}, right.
Neither $B$ nor $B'$ intersects $\Sigma$, hence neither ball contains $w$,
hence $|vo| \leq |wo|$ and $|vo'| \leq |wo'|$.
Therefore, each of $o$ and $o'$ lies either on $\Lambda$ or
on the same side of $\Lambda$ as $v$, as illustrated.
(More broadly, $oo' \subset \Vor \, v$.)
As $v \in oo'$ and $v \not\in \Lambda$,
$\Lambda$ does not intersect the relative interior of~$oo'$.
By contrast, $a$ does intersect the relative interior of $oo'$ (at $t$).
Recall that $a$'s~origin $x$ lies on $\Lambda$.
Therefore, the open ray $a$ is strictly on the same side of $\Lambda$ as
the relative interior of $oo'$, which contains $v$.
\end{proof}

Lemma~\ref{lem:pathincell}, below, shows that
under suitable sampling conditions,
every point in $\Vor|_\Sigma v \setminus \{ v \}$ lies on
one and only one radial path.
It depends on the simple observation of Lemma~\ref{lem:neighborinject}.

\begin{lemma}
\label{lem:neighborinject}
Let $v, x \in \Sigma$ be two points such that $|vx| < \xi \, \lfs(v)$.
% $|vx| \leq 0.9101 \, \lfs(v)$.
% Let $\varphi$ be the function that orthogonally projects points onto
% $v$'s tangent plane~$T_v\Sigma$.
% Let $r \subset T_v\Sigma$ be a ray with origin $v$ that passes through $y$.
% Let $N_v\Sigma$ be the line normal to $\Sigma$ at $v$, and
% let $\Pi$ be the halfplane with boundary $N_v\Sigma$ that contains $r$.
There exists an open neighborhood
$N \subset \Sigma$ of $x$ such that $\varphi|_N$ is
a homeomorphism from $N$ to its image $\varphi(N) \subset T_v\Sigma$.
% where $\varphi|_N$ denotes the restriction of $\varphi$
% to the domain $N$.
%%% JRS:  Add a definition of $\varphi|_N$?
% Moreover, there exists an open neighborhood
% $M \subset \Sigma \cap \Pi$ of $x$ such that $\varphi|_M$ is
% a homeomorphism from $M$ to its image $\varphi(M)$ on $T_v\Sigma$.
\end{lemma}

\begin{proof}
% As $|vx| \leq 0.9101 \, \lfs(v)$,
As $|vx| < \xi \, \lfs(v)$,
by Lemma~\ref{lem:abovebelow} (or Lemma~\ref{lem:nvl}),
% $\angle (n_v, n_x) \neq 90^\circ$.
% where $n_v$ and $n_x$ are
% outside-facing vectors normal to $\Sigma$ at $v$ and $x$, respectively.
% Hence
$T_x\Sigma$ is not perpendicular to $T_v\Sigma$.
It~follows from the smoothness of $\Sigma$ that
if $N$ is sufficiently small, $\varphi|_N$ is injective.
As $\varphi|_N$ is injective and
both $\varphi|_N$ and its inverse are continuous,
$\varphi|_N$ is a homeomorphism.
\end{proof}

\begin{lemma}
\label{lem:pathincell}
Consider a site $v \in V$, its restricted Voronoi cell $C = \Vor|_\Sigma v$, and
a point $x \in C \setminus \{ v \}$.
Let $r \in T_v\Sigma$ be the closed ray with origin $v$ that passes through
$\varphi(x)$.
Suppose that for every point $y \in C$, $|vy| < \xi \, \lfs(v)$, where
$\xi = \sqrt{(\sqrt{5} - 1) / 2} \doteq 0.786151$.

Then there is a unique radial path $\gamma \subset C$ such that $x \in \gamma$.
% Moreover, let $z$ be
% the endpoint of $\gamma$ on the boundary of $C$; then
% no point in $\gamma \setminus \{ z \}$ is
% in any other site's restricted Voronoi cell.
Furthermore, $\varphi(\gamma) \subseteq r$ and
$\gamma$~is the only radial path such that $\varphi(\gamma) \subseteq r$.
\end{lemma}

\begin{proof}
For every point $y \in C \setminus \{ v \}$ (including $x$),
$\varphi(y) \neq v$, as if $\varphi(y) = v$ we have a contradiction:
$|vy| < \xi \, \lfs(v)$ and $y \neq v$ imply that
$y$ is in $B$ or $B'$ and thus not in~$C$.

% Let $C = \Vor|_\Sigma v$ be a shorthand for $v$'s restricted Voronoi cell.
Define the point set $\varphi|_C^{-1}(r) = \{ y \in C : \varphi(y) \in r \}$
(the intersection of $C$ with
the closed halfplane with boundary $N_v\Sigma$, passing through $x$).
Clearly, $x \in \varphi|_C^{-1}(r)$.
Let $\gamma$ be
the connected component of $\varphi|_C^{-1}(r)$ that contains $x$.
We will show that $\gamma$ is a radial path.

As $\gamma \subset C$, % and $\varphi(x) \in r$.
for every point $y \in \gamma$, $|vy| < \xi \, \lfs(v)$ and
by Lemma~\ref{lem:neighborinject} there exists an open neighborhood
$N \subset \Sigma$ of $y$ such that $\varphi|_N$ is a homeomorphism, so
$\varphi|_{N \cap \gamma}$ is a homeomorphism
from $N \cap \gamma$ to its image $\varphi(N \cap \gamma)$. % on the ray $r$.
In other words, $\varphi|_{\gamma}$ is a local homeomorphism
from the path~$\gamma$ to its image $\varphi(\gamma)$. % on $r$.
As $\gamma$~is connected and $\varphi(\gamma)$ is embedded in the ray $r$,
$\varphi|_{\gamma}$ is a~(global) homeomorphism.
(Intuitively, the map $\varphi$ cannot cause the path to double back
on itself, so $\varphi|_{\gamma}$ is an injection.)
Therefore, $\gamma$ is a topological interval
% (i.e., a connected $1$-manifold with boundary)
or a lone point.

As $C$ is compact and $r$ is closed,
$\varphi|_C^{-1}(r)$ is compact and $\gamma$ is compact.
Let $q$ and $z$ be the endpoints of~$\gamma$,
chosen so that $|v\varphi(q)| \leq |v\varphi(z)|$;
see Figure~\ref{radialpath}.
As $\varphi|_{\gamma}$ is a homeomorphism,
$\varphi(q)$ and $\varphi(z)$ are the endpoints of $\varphi(\gamma)$.
As $\varphi(\gamma)$ contains $\varphi(x)$ and $\varphi(x) \neq v$,
$\varphi(z) \neq v$ and $z \neq v$.
We will show that $q = v$ (which implies that $\varphi(C)$ is star-shaped) and
that $z$ is on the boundary of $\Vor \, v$, thereby establishing
the first two criteria for $\gamma$ to be a radial path.

\begin{figure}
\centerline{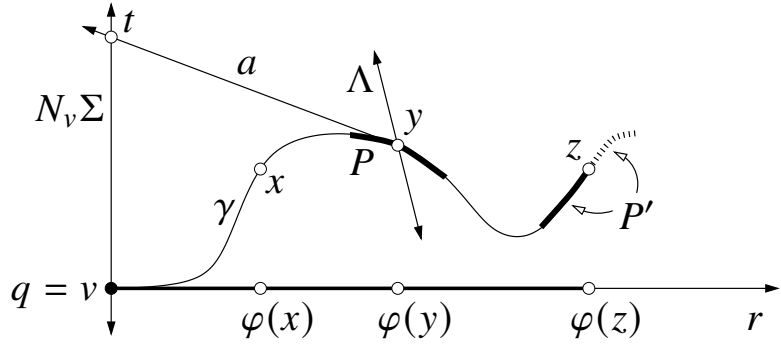}

\caption{\label{radialpath}
A radial path $\gamma$ with endpoints $v$ and $z$.
(The path $P'$ extends past $\gamma$ on $\Sigma$.)
}
\end{figure}

But first, we show that $\gamma \setminus \{ z \}$ is
in the interior of $\Vor \, v$,
the third criterion for $\gamma$ to be a radial path.
Suppose for the sake of contradiction that
a point $y \in \gamma \setminus \{ z \}$ lies on the boundary of $\Vor \, v$.
Then $y$ also lies in the restricted Voronoi cell $\Vor|_\Sigma w$ of
another site $w \in V \setminus \{ v \}$, and
$y \in \Lambda$ where $\Lambda$ is the plane that orthogonally bisects $vw$.
Clearly, $y \neq v$.
By Lemma~\ref{lem:abovebelow},
$T_y\Sigma$ intersects $N_v\Sigma$ at a lone point~$t$.
Let $a$ be the open ray $\vec{yt}$.
Observe that $a \subset T_y\Sigma$; moreover,
$a$ is tangent to $\gamma$ at~$y$, as illustrated in Figure~\ref{radialpath}.

By Lemma~\ref{lem:raybisector},
$v$ and $a$ are strictly on the same side of~$\Lambda$.
Hence if you walk along $\gamma$ from $y$ to $z$---opposite to
the direction of $a$---you enter
$w$'s side of $\Lambda$ at the instant you leave~$y$, as illustrated.
This contradicts the fact that $\gamma \in \Vor|_\Sigma v$.
So $\gamma \setminus \{ z \}$ is in the interior of $\Vor \, v$.

Let us return to the first two criteria for $\gamma$ to be a radial path.
Consider a point $y \in \gamma \setminus \{ v \}$.
By Lemma~\ref{lem:neighborinject},
there exists an open neighborhood $N \subset \Sigma$ of $y$
such that $\varphi|_N$ is a homeomorphism. % from $N$ to $\varphi(N)$;
Define $\varphi|_N^{-1}(r) = \{ p \in N : \varphi(p) \in r \}$ and
let $P$ be the connected component of $\varphi|_N^{-1}(r)$ that contains $y$,
illustrated in Figure~\ref{radialpath}.
As $\varphi(y)$ is in the interiors of $r$ and $\varphi(N)$,
$\varphi(y)$ is in the interior of $r \cap \varphi(N)$.
As $\varphi|_N$ is a homeomorphism,
$y$ is in the relative interior of $\varphi|_N^{-1}(r)$, so
$P$~is a path (topological interval) on $\Sigma$
with $y$ in its relative interior.

First, consider the case where $y$ is in the interior of $\Vor \, v$
(but $y \neq v$).
Then we can shrink the open neighborhood $N$ of $y$ so that
$N \subset \Vor \, v$ and thus $N \subset C$, and
thereby have $\varphi|_N^{-1}(r) \subseteq \varphi|_C^{-1}(r)$.
As $\gamma$~is the connected component of $\varphi|_C^{-1}(r)$ that
contains~$y$, $P \subseteq \gamma$.
Hence $y$ is not an endpoint of $\gamma$.
We have seen that $\gamma \setminus \{ z \}$ is in the interior of $\Vor \, v$;
it follows that only $v$ and $z$ can be endpoints of $\gamma$.
That is, the endpoint $q$ is either $v$ or $z$.
Moreover, as $z \neq v$ is an endpoint of $\gamma$,
$z$ is on the boundary of $\Vor \, v$ (establishing the second criterion).

Second, consider the case where $y = z$.
In this case, the path $P$ looks like $P'$ in Figure~\ref{radialpath}.
We use this case to show that $q \neq z$.
Suppose for the sake of contradiction that $q = z$; then $\gamma = \{ z \}$.
It follows that as you walk along $P$ from $z$,
you exit the restricted Voronoi cell $C$ in both directions along~$P$.
Let $a$ be the open ray $\vec{zt}$ where
$\{ t \} = T_z\Sigma \cap N_v\Sigma$ (hence $a \subset T_z\Sigma$);
$a$ is tangent to $P$ at $z$.
As $P$ exits $C$,
there exists a site $w \in V \setminus \{ v \}$ and
a plane~$\Lambda$ that orthogonally bisects $vw$ such that
$z \in \Lambda$ and $a$~enters $w$'s side of $\Lambda$ at $z$.
But this contradicts the fact that, by Lemma~\ref{lem:raybisector},
$v$ and $a$ are strictly on the same side of $\Lambda$.
Thus $q \neq z$, so $q = v$.

Therefore, the endpoints of $\gamma$ are $v$ and $z$,
$z$ is on the boundary of $\Vor \, v$ and
$\gamma \setminus \{ z \}$ is in its interior, and
$\varphi|_\gamma$ is a homeomorphism from $\gamma$ to $v\varphi(z)$.
By definition, $\gamma$ is a radial path.

To see that $\gamma$ is the {\em only} radial path containing $x$, observe that
every radial path containing~$x$ is a subset of $\varphi|_C^{-1}(r)$, and
moreover is a subset of $\gamma$ (because a radial path is connected).
No strict subset of $\gamma$ can be a radial path, because
every connected strict subset of $\gamma$ is missing
either $v$ or a point on the boundary of $\Vor \, v$.

The reasoning of this proof holds equally well if we replace $x$ with
any other point $x' \in \varphi|_C^{-1}(r) \setminus \{ v \}$, so
every connected component of $\varphi|_C^{-1}(r)$ contains $v$.
Therefore, $\varphi|_C^{-1}(r)$ is connected.
Hence $\gamma = \varphi|_C^{-1}(r)$ and
no radial path besides $\gamma$ has its projection on $r$.
\end{proof}

It follows that we can decompose $\Vor|_\Sigma v$ into radial paths.
% (See Appendix~\ref{pathsproof} for a proof.)
Hence $\varphi(\Vor|_\Sigma v)$ is star-shaped.
It also follows that the orthogonal projection $\varphi$,
restricted to $\Vor|_\Sigma v$, is a homeomorphism.

\begin{lemma}
\label{lem:pathsincell}
% Suppose the conditions specified in Lemma~\ref{lem:pathincell} hold.
% Let $v \in V$ be a site.
Consider a site $v \in V$ and
its restricted Voronoi cell $C = \Vor|_\Sigma v$.
Suppose that for every point $y \in C$, % $y \in \Vor|_\Sigma v$,
$|vy| < \xi \, \lfs(v)$, where
$\xi = \sqrt{(\sqrt{5} - 1) / 2} \doteq 0.786151$.

Let $\Gamma$ be the set of all radial paths for all points in
$C \setminus \{ v \}$. % $\Vor|_\Sigma v \setminus \{ v \}$.

Then $\bigcup_{\gamma \in \Gamma} \gamma = C$; hence
$\varphi(C)$ is star-shaped.
Moreover, no two paths in~$\Gamma$ share a common point besides $v$, and
there is a one-to-one correspondence between paths in $\Gamma$ and
points where $\Sigma$ intersects the boundary of $\Vor \, v$.
% Moreover, no other restricted Voronoi cell intersects the interior of~$C$.

Moreover, $\varphi|_C$ is a homeomorphism from $C$ to
its image $\varphi(C)$ on~$T_v\Sigma$.
\end{lemma}

\begin{proof}
By Lemma~\ref{lem:pathincell}, for each point $x \in C \setminus \{ v \}$,
there is a unique radial path $\gamma \subset C$ such that $x \in \gamma$.
Every radial path contains $v$.
Hence $\bigcup_{\gamma \in \Gamma} \gamma = C$.
As each $x \in C \setminus \{ v \}$ lies on only one radial path,
no two paths in~$\Gamma$ share a common point besides~$v$.
By definition,
each radial path contains exactly one point on the boundary of $\Vor \, v$.
Hence there is a one-to-one correspondence between radial paths and
points where $\Sigma$ intersects the boundary of $\Vor \, v$.

Let us see that $\varphi|_C$ is an injection.
Consider two points $x, y \in C$ such that $\varphi(x) = \varphi(y)$;
we will see that $x = y$.
For every radial path $\gamma \in \Gamma$,
$\varphi|_\gamma$ is a homeomorphism by the definition of radial path, so
if $x$ and $y$ lie on the same radial path, then $x = y$.
%(This includes the case where $\varphi(x) = \varphi(y) = v$, thus $x = y = v$.)
If $x$~and $y$ lie on distinct radial paths, then $x = y = v$, because
by Lemma~\ref{lem:pathincell},
for any two distinct radial paths $\gamma_1, \gamma_2 \in \Gamma$,
$\varphi(\gamma_1)$ and $\varphi(\gamma_2)$ lie on
two distinct rays with origin $v$.

% The orthogonal projection $\varphi$ is continuous.
As $C$ is compact and $\varphi|_C$ is injective and continuous,
$\varphi|_C$ is a homeomorphism~\cite{sutherland09}.
\end{proof}

Lemma~\ref{lem:pathsincell} shows that
% $\varphi|_{\Vor|_\Sigma v}$ is a homeomorphism from $\Vor|_\Sigma v$ to
$\Vor|_\Sigma v$ is homeomorphic to
its image $I_v = \varphi(\Vor|_\Sigma v)$ on $T_v\Sigma$, but
what is the shape of $I_v$?
We obtain a homeomorphism from $I_v$ to a closed unit disk on~$T_v\Sigma$
by simply scaling each line segment $\varphi(\gamma)$ to have unit length.
Thus we arrive at this section's main theorem, Theorem~\ref{thm:homeocell},
which states that $\Vor|_\Sigma v$ is a topological closed disk.

Let $E = \{ \varphi(\gamma) : \gamma \in \Gamma \}$ be the set of
the orthogonal projections of the radial paths onto $T_v\Sigma$.
Then we can write $I_v = \bigcup_{e \in E} e$,
a decomposition of $I_v$ into line segments with endpoint~$v$,
no two leaving $v$ in the same direction % by Lemma~\ref{lem:pathincell}
(but every direction is represented).

% For each line segment $e \in E$, let $l(e)$ be the length of $e$.
% For every point $x \in I_v \setminus \{ v \}$,
% let $e_x$ denote the unique line segment in $E$ that contains $x$, and
% let $l(x) = l(e_x)$.
For every point $x \in I_v \setminus \{ v \}$,
let $l(x)$ be the length of the unique line segment in $E$ that contains $x$.
Thus $l$ is a function over the domain $I_v \setminus \{ v \}$,
but $l(v)$ is not defined.
The forthcoming Lemma~\ref{lem:slicecontinuous} shows that
that $l$ is continuous.
This is a consequence of two facts:  $I_v$~is a compact point set and
every point on a line segment $e \in E$ except one endpoint is
in the interior of $I_v$
(where the {\em interior} is defined with respect to $T_v\Sigma$).

Let $\chi: I_v \rightarrow T_v\Sigma$ map
each line segment in $E$ to a line segment with unit length
(while preserving its direction)---specifically,
$\chi(x) = v + \frac{1}{l(x)} (x - v)$ for $x \neq v$ and
$\chi(v) = v$.
Then $\chi$~is continuous as $l$ is continuous and positive.
% As $\chi$ is an injective, continuous map over a compact domain,
% $\chi$ is a homeomorphism from $I_v$ to
% a unit disk~\cite[Corollary~13.27]{sutherland09}. % centered at $v$.
As $\chi$ is a bijective, continuous map with a continuous inverse,
$\chi$ is a homeomorphism from $I_v$ to a closed unit disk.

The (tedious) proof of $l$'s continuity, Lemma~\ref{lem:slicecontinuous},
follows the theorem that uses it.

\begin{theorem}
\label{thm:homeocell}
Consider a site $v \in V$ and its restricted Voronoi cell $C = \Vor|_\Sigma v$.
Suppose that for every point $y \in \Vor|_\Sigma v$,
$|vy| < \xi \, \lfs(v)$, where
$\xi = \sqrt{(\sqrt{5} - 1) / 2} \doteq 0.786151$.

Then $\chi \circ \varphi|_C$~is a homeomorphism
from $\Vor|_\Sigma v$ to a closed unit disk on~$T_v\Sigma$.
\end{theorem}

\begin{proof}
%%%CUT when moving out of appendix.
Recall that $\chi(x) = v + \frac{1}{l(x)} (x - v)$ for $x \neq v$ and
$\chi(v) = v$.
By Lemma~\ref{lem:slicecontinuous}, $l$ is continuous; $l$ is also positive.
Hence $\chi$ too is continuous over $I_v \setminus \{ v \}$.
The map $\chi$ is continuous at $v$ as well, because
$\chi(x)$ converges to $v$ in the limit as $x$ approaches $v$.

It is easy to see that $\chi$ is injective.
Consider two points $x, y \in I_v$ such that $\chi(x) = \chi(y)$.
If $\chi(x) = \chi(y) = v$, then $x = y = v$.
Otherwise, $x$ and $y$ must lie on the same segment $e \in E$ because
$\vec{vx} = \vec{v\chi}(x) = \vec{v\chi}(y) = \vec{vy}$.
But $\chi|_e$ is clearly injective, so $x = y$.
Thus $\chi$ is injective.

As $I_v$ is compact and $\chi$ is injective and continuous,
$\chi$ is a homeomorphism from $I_v$ to
its image $\chi(I_v)$~\cite[Corollary 13.27]{sutherland09}.
As $\chi(I_v)$ is a union of closed, unit-length line segments emanating
from~$v$ in all directions on $T_v\Sigma$, $\chi(I_v)$ is a closed unit disk.

By Lemma~\ref{lem:pathsincell},
$\varphi|_C$ is a homeomorphism from $C$ to $I_v$, so
$\chi \circ \varphi|_C$ is a homeomorphism from $C$ to the closed unit disk.
\end{proof}

\begin{lemma}
\label{lem:slicecontinuous}
% Suppose the condition specified in Lemmas~\ref{lem:pathincell} holds.
Consider a site $v \in V$, its restricted Voronoi cell $C = \Vor|_\Sigma v$,
the cell's image $I_v = \varphi(C)$ on $T_v\Sigma$, and
the function $l : I_v \setminus \{ v \} \rightarrow \R$.
Suppose that for every point $y \in C$, % $y \in \Vor|_\Sigma v$,
$|vy| < \xi \, \lfs(v)$, where
$\xi = \sqrt{(\sqrt{5} - 1) / 2} \doteq 0.786151$.
Then $l$~is continuous. % over $I_v \setminus \{ v \}$.
\end{lemma}

\begin{proof}
We will see that
for every $x \in I_v \setminus \{ v \}$ and every $\delta > 0$,
there is a neighborhood $N \subset I_v \setminus \{ v \}$ of $x$ such that
for every point $y \in N$, $l(y) \in (l(x) - \delta, l(x) + \delta)$.
Continuity follows.

Let $e \in E$ be the line segment that contains $x$.
One endpoint of $e$ is $v$; let $z$ be the other endpoint.
Observe that $l(x) = |vz|$, the length of $e$.

\begin{figure}
\centerline{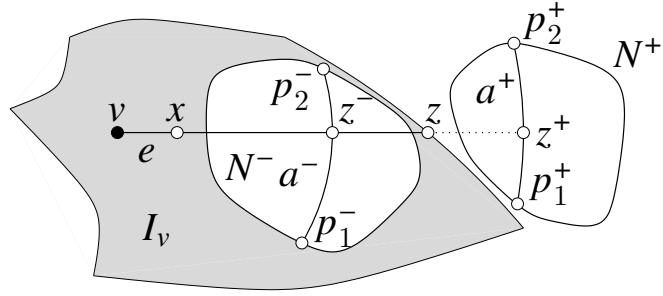}

\caption{\label{lcontinuous}
As the arc $a^+$ does not intersect $I_v$ and
the arc $a^-$ is in the interior of $I_v$,
every line segment $e' \in E$ such that $\angle (e, e') <
\min \{ \angle xvp_1^+, \angle xvp_2^+, \angle xvp_1^-, \angle xvp_2^- \}$
has length between $|vz^-|$ and~$|vz^+|$.
}
\end{figure}

Let $z^-$ be the point on $e$ at a distance of $\delta$ from $z$,
as illustrated in Figure~\ref{lcontinuous}.
We assume $\delta < l(x)$; otherwise,
place $z^-$ anywhere on $e \setminus \{ v, z \}$.
Let $z^+ \in T_v\Sigma$ be the point at a distance of $\delta$ from $z$
in the other direction from $z$, so that $z$ is between $z^-$ and $z^+$,
as illustrated.
As $C$~is compact, $I_v = \varphi(C)$ is compact.
% (with respect to the space $T_v\Sigma$).
As $z^+ \not\in I_v$, there is an open neighborhood
$N^+ \subset T_v\Sigma$ of~$z^+$ that is disjoint from $I_v$.

There is also an open neighborhood $N^- \subset T_v\Sigma$ of $z^-$ such that
$N^- \subset I_v$ but
$N^-$ is disjoint from the outer endpoints of the line segments in $E$;
the reasons are a bit more complicated.
Let $\gamma \in \Gamma$ be the radial path such that
$e = \varphi(\gamma)$.
As $\varphi|_{\gamma}$ is a homeomorphism from $\gamma$ to $e$ and
$z^-$ lies on $e$ but is not an endpoint of $e$,
the point $z^C = \varphi|_{\gamma}^{-1}(z^-)$ lies on~$\gamma$ but
is not an endpoint of~$\gamma$.
Hence $z^C$~lies in the interior of $\Vor \, v$
(by the definition of radial path).
Let $N^C \subset \Sigma$ be an open neighborhood of $z^C$ that lies
entirely in the interior of $\Vor \, v$; hence $N^C \subset C$.
Let $N^- = \varphi(N^C)$, thus $N^- \subset I_v$.
By Lemma~\ref{lem:pathsincell}, $\varphi|_C$ is a homeomorphism, so
$N^- \subset T_v\Sigma$ is an open neighborhood of $z^-$.
By the definition of radial path,
$\varphi|_C$ maps every point in the interior of $\Vor \, v \setminus \{ v \}$
to the relative interior of a line segment in $E$, so
$N^-$ does not contain the outer endpoint of any line segment in $E$.

Let $a^+ \subset N^+$ be an arc with center $v$ and
with $z^+$ in its relative interior, as illustrated.
Likewise, let $a^- \subset N^-$ be an arc with center $v$ and
with $z^-$ in its relative interior.
Let $p_1^+$ and $p_2^+$ be the endpoints of $a^+$ and
let $p_1^-$ and $p_2^-$ be the endpoints of~$a^-$.
Let \mbox{$\theta = \min \{ \angle xvp_1^+, \angle xvp_2^+,
\angle xvp_1^-, \angle xvp_2^- \}$}, which is positive.

As $a^+$ is disjoint from $I_v$, every line segment $e' \in E$
such that $\angle (e, e') < \theta$ has length less than $l(x) + \delta$.
As $a^- \subset I_v$ and $a^-$ does not intersect
the outer endpoint of any line segment in~$E$,
every line segment $e' \in E$ such that $\angle (e, e') < \theta$ has
length greater than $l(x) - \delta$.
Therefore, for every point $y \in I_v \setminus \{ v \}$ such that
$\angle xvy < \theta$, $l(y) \in ( l(x) - \delta, l(x) + \delta)$.

Let $N = \{ y \in I_v \setminus \{ v \} : \angle xvy < \theta \}$
and observe that $x$ is in the interior of $N$
(where {\em interior} is defined relative to
the topological space $I_v \setminus \{ v \}$;
note that $x$ might be on the boundary of~$I_v$
relative to the plane $T_v\Sigma$, but
the continuity of $l$ over $I_v \setminus \{ v \}$ depends only on
open sets in the space $I_v \setminus \{ v \}$).
For every point $y \in N$, $l(y) \in (l(x) - \delta, l(x) + \delta)$.
As we can identify such a neighborhood for
every $x \in I_v \setminus \{ v \}$ and every $\delta > 0$, $l$ is continuous.
\end{proof}

% It is notable that the condition $|vy| < \xi \, \lfs(w)$ of
% Lemma~\ref{lem:pathincell} and Theorem~\ref{thm:homeocell} implies,
% by the Normal Variation Lemma (Lemma~\ref{lem:nvl}), that
% $\angle (n_v, n_y) < \eta(\xi) = 60^\circ$.
% The sixty-degree bound is exact.
% (We think it is just a coincidence that that number comes out so cleanly.
% The bound is probably not tight.)

% See Appendix~\ref{2ballproof}.
If we impose the condition of Theorem~\ref{thm:homeocell} on
{\em all} the restricted Voronoi cells,
% (or the $\epsilon$-sample condition of
% Corollary~\ref{cor:homeocelleps}, below),
every connected component of $\Sigma$ has at least six sites on it.
% (Lemma~\ref{lem:sixsites} follows easily from
% Lemmas~\ref{lem:pathincell} and~\ref{lem:nvl}, but
% there isn't quite enough space here for the proof;
% see the full-length paper~\cite{shewchuk26b}.)

\begin{lemma}
\label{lem:sixsites}
Let $V$ be an $\epsilon$-Voronoi sample of $\Sigma$ for some
$\epsilon < \xi = \sqrt{(\sqrt{5} - 1) / 2} \doteq 0.786151$.
%Let $V \subset \Sigma$ be a nonempty, finite set of points (sites) on $\Sigma$.
%Suppose that for every site $v \in V$ and every point $x \in \Vor|_\Sigma v$,
%$|vx| < \xi \, \lfs(v)$,
%where $\xi = \sqrt{(\sqrt{5} - 1) / 2} \doteq 0.786151$.
% (Alternatively, suppose that $V$ is a $\epsilon$-sample of
% $\Sigma$ for $\epsilon < \frac{\xi}{\xi + 1} \doteq 0.440137$.)
% Suppose, moreover, that for every three distinct sites $p, p', p'' \in V$,
% the intersection of their restricted Voronoi cells contains either
% one point (i.e., a restricted Voronoi vertex) or zero.
%
%%%UNCUT
% Then ... at least
Every connected component of $\Sigma$ has
at least six sites and six restricted Voronoi cells on it.
% Then every restricted Voronoi edge has two distinct vertices,
% every restricted Voronoi cell has at least two restricted Voronoi edges and
% at least two restricted Voronoi vertices on its boundary, and
% every connected component of $\Sigma$ has at least
% six restricted Voronoi cells, six restricted Voronoi edges, and
% two restricted Voronoi vertices on it.
\end{lemma}

\begin{proof}
% By Theorem~\ref{thm:homeocell} (or Corollary~\ref{cor:homeocelleps}),
Lemma~\ref{lem:pathincell} implies that
every restricted Voronoi cell in $\Vor|_\Sigma V$ is connected.
% homeomorphic to a disk.
Therefore, each cell lies on just one connected component of $\Sigma$.
% Every point $x$ in $v$'s restricted Voronoi cell $\Vor|_\Sigma v$
% lies on the same connected component of $\Sigma$ as $v$ does, because
% if they lay on different connected components, then
% the line segment $vx$ would intersect the medial axis, contradicting
% the fact that $|vx| < \lfs(v)$.
By the Normal Variation Lemma (Lemma~\ref{lem:nvl}),
for every site $v \in V$ and every point $x \in \Vor|_\Sigma v$,
$\angle (n_v, n_x) < \eta(\xi) = 60^\circ$.

Let $\mathring{\Sigma}$ be a connected component of $\Sigma$.
% For any unit vector $u$ on the unit sphere, let $y$ be
% a point on $\mathring{\Sigma}$ that is most extreme in the direction $u$.
% As $\mathring{\Sigma}$ is a smooth surface without boundary in $\R^3$,
% $u$ is normal to $\mathring{\Sigma}$ at $y$ and
% oriented to the outside of $\mathring{\Sigma}$; that is,
% the unit vector $n_y = u$ is an outside-facing normal vector at $y$.
% It follows that the outside-facing unit normal vectors on $\mathring{\Sigma}$
% constitute the entire sphere of directions.
It is well known that the outside-facing unit normal vectors on
a closed, smooth surface $\mathring{\Sigma}$ constitute the entire unit sphere.
A restricted Voronoi cell $\Vor|_\Sigma v$ can contain
only points on $\Sigma$
whose outside-facing normals are less than $60^\circ$ from~$n_v$.
At least six sites are required on a sphere so that
every point on the sphere is less than~$60^\circ$ from one of the sites;
five do not suffice~\cite{cohn26}.
% (See {\em Coverings by points on a sphere},
% R.~H.\ Hardin, N.~J.~A.\ Sloane, and W.~D.\ Smith, February 1994.
(The six points where the coordinate axes intersect the unit sphere suffice.)
Hence there are at least six sites and six restricted Voronoi cells
on~$\mathring{\Sigma}$.
%
% The same reasoning applies to every connected component of~$\Sigma$.
The same applies to every connected component of~$\Sigma$.
\end{proof}

To apply Theorem~\ref{thm:homeocell} to $\epsilon$-samples,
observe that $0.44$-samples are $0.786$-Voronoi samples.
% the sampling condition of Theorem~\ref{thm:homeocell}.

% The following simple lemma is implicit in Amenta and Bern~\cite{amenta99b} and
% explicit in Amenta, Choi, Dey, and Leekha~\cite{amenta02}.

\begin{lemma}[Feature Translation Lemma~\cite{amenta02,dey07}]
\label{lem:ftl}
Let $\Sigma \subset \R^3$ be a smooth surface and
let $p, q \in \Sigma$ be points such that
$|pq| \leq \epsilon \, \lfs(p)$ for some $\epsilon < 1$.
Then
\[
\lfs(p) \leq \frac{1}{1 - \epsilon} \lfs(q)
\hspace{.2in}  \mbox{and}  \hspace{.2in}
|pq| \leq \frac{\epsilon}{1 - \epsilon} \lfs(q).
\]
\end{lemma}

\begin{proof}
By the definition of the local feature size,
there is a medial axis point $m$ such that $|qm| = \lfs(q)$.
By the Triangle Inequality,
\[
\lfs(p) \leq |pm| \leq |pq| + |qm| \leq \epsilon \, \lfs(p) + \lfs(q).
\]
Rearranging terms gives $\lfs(p) \leq \lfs(q) / (1 - \epsilon)$.
The second claim follows immediately.
\end{proof}

\begin{corollary}
\label{cor:homeocelleps}
Let $V$ be an $\epsilon$-sample of $\Sigma$ for
$\epsilon < \frac{\xi}{\xi + 1} \doteq 0.440137$.

Then every restricted Voronoi cell in $\Vor|_\Sigma V$ is
homeomorphic to a closed disk (closed disk property~A),
$\Sigma$~does not intersect any $2$-face of $\Vor \, V$ tangentially
(generic intersection property~E),
every connected component of $\Sigma$ has
at least six sites and six restricted Voronoi cells on it, and
no restricted Voronoi cell intersects
%%%UNCUT
the interior of another. % restricted Voronoi cell.
\end{corollary}

\begin{proof}
Consider any site $v \in V$ and
any point $x \in \Vor|_{\Sigma} v$.
As $V$ is an $\epsilon$-sample, $|vx| \leq \epsilon \, \lfs(x)$.
By the Feature Translation Lemma (Lemma~\ref{lem:ftl}),
$|vx| \leq \frac{\epsilon}{1 - \epsilon} \, \lfs(v) < \xi \, \lfs(v)$.

The first three claims follow from Theorem~\ref{thm:homeocell} and
Lemmas~\ref{lem:raybisector} and~\ref{lem:sixsites}, respectively.
The final claim follows from Lemma~\ref{lem:pathsincell} because
every point in the interior of $\Vor|_{\Sigma} v$
is in the interior of $\Vor \, v$ (by the definition of radial path), and
cannot be shared with another cell.
\end{proof}

\section{Restricted Voronoi vertices are lone points}
\label{0ball}

% Assuming $\Sigma$ is a $2$-manifold,
A restricted Voronoi vertex is
a nonempty intersection of $\Sigma$ with an edge in $\Vor \, V$.
However, without suitable sampling conditions,
such an intersection might contain many points, even infinitely many.
% (For instance, there can be infinitely many when
% the intersection of $\Sigma$ with
% an edge of the full-dimensional Voronoi diagram is a line segment.)
This section describes sampling conditions that guarantee
closed ball property~C:
an intersection of three distinct restricted Voronoi cells contains
at most one point, thereby justifying the name ``vertex.''
Generic intersection property F comes as a~byproduct.
Lemma~\ref{lem:onlyonevertex} shows that a $0.49$-Voronoi sample suffices,
and Corollary~\ref{cor:onlyonevertexeps} shows that a $0.33$-sample suffices
(which follows from Lemmas~\ref{lem:onlyonevertex} and~\ref{lem:ftl}).
% the Feature Translation Lemma, Lemma~\ref{lem:ftl}).
% We postpone both proofs to the full-length paper~\cite{shewchuk26b}, as
% restricted Voronoi vertices are not the bottleneck limiting our
% homeomorphism theorems' sample bounds.
% (Restricted Voronoi edges are the bottleneck; see Section~\ref{1ball}.)
%%%CUT
Note that Cheng et al.~\cite{cheng12} prove Lemma~\ref{lem:onlyonevertex} for
a $0.15$-Voronoi sample.

We need four preliminary lemmas to prove Lemma~\ref{lem:onlyonevertex}.
The first three are simple statements in surface sampling theory,
proved in other sources.

\begin{lemma}[Feature Ball Lemma~\cite{dey07,cheng12}]
\label{lem:fbl}
If a geometric closed $d$-ball $B$ intersects
a \mbox{$k$-manifold} $\Sigma \subset \R^d$ without boundary
at more than one point and either
(i)~$\Sigma \cap B$ is not a topological $k$-ball or
(ii)~$\Sigma \cap \Bd B$ is not a topological $(k - 1)$-sphere, then
$B$ contains a medial axis point.
\end{lemma}

As $\Sigma$ is a closed surface, it cuts space into two pieces,
a bounded {\em inside} piece and an unbounded {\em outside} piece.
For a point $x \in \Sigma$, let $n_x$ denote the outside-facing vector
normal to $\Sigma$ at $x$.
Let $\angle (n_x, n_y) \in [0^\circ, 180^\circ]$ denote
the angle separating two vectors.

\begin{lemma}[Normal Variation Lemma~\cite{khoury19}]
\label{lem:nvl}
%%% UNCUT
% Let $\Sigma \subset \R^3$ be a smooth, closed surface.
Consider two points $p, q \in \Sigma$ and let $\delta = |pq| / \lfs(p)$.
% Let $n_p$ and $n_q$ be
% outside-facing vectors normal to $\Sigma$ at $p$ and $q$, respectively.
If $\delta < \sqrt{4 \sqrt{5} - 8} \doteq 0.971736$, then
$\angle (n_p, n_q) \leq \eta(\delta)$ where
\[
\eta(\delta) = \arccos \left( 1 - \frac{\delta^2}{2 \sqrt{1 - \delta^2}}
                              \right)
\approx
\delta + \frac{7}{24} \delta^3 + % \frac{123}{640} \delta^5 +
% \frac{1\mbox{,}083}{7\mbox{,}168} \delta^7 + O(\delta^9).
O(\delta^5).
% \label{codim1bound}
\]
\end{lemma}

\begin{lemma}[Triangle Normal Lemma~\cite{khoury19}]
\label{lem:tnl}
%%%UNCUT
% Let $\Sigma \subset \R^3$ be a smooth, closed $2$-manifold.
Let $\tau$ be a triangle whose vertices lie on $\Sigma$.
% Let $r$ be $\tau$'s circumradius.
Let $r$ be the radius of $\tau$'s circumscribing circle.
Let $v$ be a vertex of $\tau$ and let $\phi$ be $\tau$'s plane angle at $v$.
Let $n_\tau$ be a vector normal to $\tau$.
Let $\aff \tau$ denote the affine hull of $\tau$.
% directed so that $\angle(n_\tau, n_v) \leq 90^\circ$.
Then
\[
\sin \angle(n_\tau, n_v) = \sin \angle(\aff \tau, T_v\Sigma) \leq
\frac{r}{\lfs(v)} \max \left\{ \cot \frac{\phi}{2}, 1 \right\}.
% \label{tnlbound1}
\]
% (Note that the argument $\cot \frac{\phi}{2}$ dominates if $\phi$ is acute and
% the argument $1$ dominates if $\phi$ is obtuse.)
% In particular, if $v$ is the vertex at $\tau$'s largest plane angle (so
% $\phi \geq 60^\circ$) and $r < \lfs(v) / \sqrt{3} \doteq 0.577 \, \lfs(v)$,
% then
% \[
% \angle(N_\tau, N_v\Sigma) = \angle(\aff \tau, T_v\Sigma) \leq
% \arcsin \frac{\sqrt{3}R}{\lfs(v)}.
% % \label{tnlbound2}
% \]
\end{lemma}

\begin{lemma}
\label{lem:projinject}
Let $\Sigma$ be a smooth, closed surface in $\mathbb{R}^3$.
Let $Q \subset \mathbb{R}^3$ be a closed, convex point set.
Let $\widetilde{\Sigma} = Q \cap \Sigma$ and
suppose that $\widetilde{\Sigma}$ is connected.
Let $A \subset \mathbb{R}^3$ be a plane, and let $\psi$ be
the map that orthogonally projects $\R^3$ onto $A$.
% (The direction of that projection is parallel to $n_A$).
Suppose that for every point $p \in \widetilde{\Sigma}$,
the vector $n_p$ normal to $\Sigma$ at $p$ is not parallel to $A$.
Then the restricted projection $\psi|_{\widetilde{\Sigma}}$ is
% restriction of $\psi$ to $\widetilde{\Sigma}$ is
a homeomorphism from $\widetilde{\Sigma}$ to
its image $\psi(\widetilde{\Sigma})$ on $A$.
\end{lemma}

\begin{proof}
Most of the effort is to show that $\psi|_{\widetilde{\Sigma}}$ is injective.
Suppose for the sake of contradiction that
two distinct points $x, z \in \widetilde{\Sigma}$ have $\psi(x) = \psi(z)$.
All the points in $\widetilde{\Sigma}$ that map to $\psi(x)$,
including $x$ and $z$, lie on
a common line $\ell$, perpendicular to $A$ and passing through $x$.
% By the Feature Ball Lemma (Lemma~\ref{lem:fbl}), $\widetilde{\Sigma}$ is
% a topological disk, so every point in $\widetilde{\Sigma}$ inherits
% a normal vector from $\Sigma$ and $\widetilde{\Sigma}$ is smooth.
As $Q$ is convex, $xz \subseteq Q$, so
every point in $\Sigma \cap xz$ is in $\widetilde{\Sigma}$.
By assumption, $n_p$~is not parallel to $A$
for any $p \in \widetilde{\Sigma}$, so
$\ell$ does not intersect $\Sigma$ tangentially at
any point in~$\widetilde{\Sigma}$, and
$xz$ does not intersect $\Sigma$ tangentially at any point.
Let $y$ be the point in $\Sigma \cap xz$ that is nearest $x$ but is not~$x$.
(The point $y$ might be $z$, or there might be a point on $xz$ closer to $x$.)
No point in $\Sigma$ is between $x$ and~$y$; that is,
among the points in $\Sigma \cap \ell$,
$x$ and $y$ are successive along $\ell$.
Observe that $y \in \widetilde{\Sigma}$.

Let $n_A$ be a vector normal to $A$ (and parallel to~$\ell$),
directed so that $\angle (n_A, n_x) < 90^\circ$.
As we walk along $\ell$, if $x$ represents a transition from inside to outside,
then $y$ represents a transition from outside to inside (and vice versa), so
$\angle (n_A, n_y) > 90^\circ$.

% By the Feature Ball Lemma (Lemma~\ref{lem:fbl}),
As $\widetilde{\Sigma}$ is connected,
there is a path $\gamma \subset \widetilde{\Sigma}$ connecting $x$ to $y$.
As $\Sigma$ is smooth, the outside-facing normal vector $n_p$
varies continuously for $p \in \gamma$.
Therefore, there is a point $q \in \gamma$ for which
$\angle (n_A, n_q) = 90^\circ$, which contradicts the assumption that
for every $p \in \widetilde{\Sigma}$, $n_p$~is not parallel to $A$.
Hence $\psi|_{\widetilde{\Sigma}}$ is injective.

As $Q$ is closed and $\Sigma$ is compact,
$\widetilde{\Sigma} = Q \cap \Sigma$ is compact.
As $\widetilde{\Sigma}$ is compact and $\psi|_{\widetilde{\Sigma}}$ is
injective and continuous, $\psi|_{\widetilde{\Sigma}}$ is
a homeomorphism~\cite[Corollary~13.27]{sutherland09}.
\end{proof}

\begin{lemma}
\label{lem:onlyonevertex}
% Let $\Sigma \subset \R^3$ be a smooth, closed surface.
% Let $V \subset \Sigma$ be a finite set of points (sites) on $\Sigma$.
Consider three distinct sites $v, v', v'' \in V$ and
the triangle $\tau = \triangle vv'v''$, where
$v$~is the vertex at $\tau$'s largest plane angle.
Let $f = \Vor|_\Sigma v \cap \Vor|_\Sigma v' \cap \Vor|_\Sigma v''$,
the restricted Voronoi face dual to $\tau$.
% generated by the sites $v$, $v'$, and $v''$.
Let $\ell_\tau \subset \R^3$ be the line containing
all points equidistant to the sites $v$, $v'$, and~$v''$
(thus $f \subset \ell_\tau$).
Suppose that for every point $y \in f$, $|vy| < \kappa \, \lfs(v)$, where
$\kappa$~is the positive real root of
$\kappa^4 = 4 (1 - \kappa^2) (1 - \sqrt{3} \kappa)^2$,
with approximate value $\kappa \doteq 0.495683$.

Then $f$ contains at most one point (i.e., a restricted Voronoi vertex).
Moreover, if $f$~contains a point,
$\Sigma$ is not tangent to $\ell_\tau$ at that point.
\end{lemma}

\begin{proof}
% Define $r$, $c$, $n_\tau$, $A$, and $\psi$ as in
% the proof of Lemma~\ref{lem:onlyonevertexeps}.
% (As in that proof, $v$, $v'$, and $v''$ cannot be collinear when
% $f \neq \emptyset$.)
If $f = \emptyset$ then the result follows.
Every point in $f$ is equidistant to the three sites $v$, $v'$, and $v''$, so
$f \subset \ell_\tau$.
If the three sites are collinear, then $f = \ell_\tau = \emptyset$ and
the result follows.

Henceforth, assume that $f \neq \emptyset$ and
the three sites are not collinear.
Then the line $\ell_\tau$ is perpendicular to $\tau$ and
passes through $\tau$'s circumcenter.
Let $r$ and $c$ be the circumradius and circumcenter of $\tau$, respectively.
Let $y$ be any point in $f$.
Observe that $c$ is the point on $\ell_\tau$ closest to the vertices of $\tau$,
so $r = |vc| \leq |vy| < \kappa \, \lfs(v)$.
Let $n_\tau$ be a vector normal to $\tau$ (and parallel to $\ell_\tau$)
directed so that $\angle (n_v, n_\tau) \leq 90^\circ$.
% As $r < \kappa \, \lfs(v)$,
By the Triangle Normal Lemma (Lemma~\ref{lem:tnl}),
$\angle (n_v, n_\tau) < \arcsin(\sqrt{3} \kappa)$.

% The present proof is similar to the proof of
% Lemma~\ref{lem:onlyonevertexeps},
% but we have a weaker bound on $|vp|$ for points $p \in \Vor|_\Sigma v$ and
% a stronger bound on $r / \lfs(v)$.
Let $Q$ be the open ball of radius $\kappa \, \lfs(v)$ centered at $v$
and let $\widetilde{\Sigma} = Q \cap \Sigma$.
As $Q$ does not intersect $\Sigma$'s medial axis,
$\widetilde{\Sigma}$ is a topological open disk by
the Feature Ball Lemma (Lemma~\ref{lem:fbl}).
By assumption, $f \subseteq Q$, so $f \subseteq \widetilde{\Sigma}$.

We show that for every point $p \in \widetilde{\Sigma}$,
$\angle (n_\tau, n_p) < 90^\circ$.
Let $p$ be a point in $\widetilde{\Sigma}$.
As $|vp| < \kappa \, \lfs(v)$,
by the Normal Variation Lemma (Lemma~\ref{lem:nvl}),
% $\angle (n_v, n_p) < \eta(\kappa)$ where $\eta(\delta) =
% \arccos \left( 1 - \frac{\delta^2}{2 \sqrt{1 - \delta^2}} \right)$.
% Observe that $\eta(\kappa) =
$\angle (n_v, n_p) < \eta(\kappa) =
\arccos \left( 1 - \frac{\kappa^2}{2 \sqrt{1 - \kappa^2}} \right) =
\arccos \left( 1 - (1 - \sqrt{3} \kappa) \right) = \arccos (\sqrt{3} \kappa)$.
Hence $\angle (n_\tau, n_p) \leq \angle (n_v, n_\tau) + \angle (n_v, n_p) <
\arcsin(\sqrt{3} \kappa) + \arccos (\sqrt{3} \kappa) = 90^\circ$.
(It is interesting to compare the magnitudes of the parts:
$\angle (n_v, n_\tau) < 59.16^\circ$ and $\angle (n_v, n_p) < 30.85^\circ$.)

Consider the first claim:  that $f$ contains at most one point.
Let $A$ be the affine hull of~$\tau$ (a plane).
Let $\psi$ be the continuous map that orthogonally projects
$\mathbb{R}^3$ onto~$A$.
The direction of that projection is parallel to $\ell_\tau$.
As $f \subset \ell_\tau$, for every pair of points $y, z \in f$,
$\psi(y) = \psi(z)$.
As $\angle (n_\tau, n_p) < 90^\circ$ for all $p \in \widetilde{\Sigma}$,
by Lemma~\ref{lem:projinject},
the orthogonal projection $\psi|_{\widetilde{\Sigma}}$ is a homeomorphism.
As $f \subseteq \widetilde{\Sigma}$, there cannot be
two distinct points $y, z \in f$ such that $\psi(y) = \psi(z)$.
Therefore, $f$ contains at most one point.

If $f$ contains a point $p$, then
the fact that $\angle (n_\tau, n_p) \neq 90^\circ$ implies the second claim:
$\Sigma$ is not tangent to $\ell_\tau$ at $p$.
\end{proof}

\begin{corollary}
\label{cor:onlyonevertexeps}
% Let $\Sigma \subset \R^3$ be a smooth, closed surface.
Let $V$ be an $\epsilon$-sample of $\Sigma$ for
$\epsilon < \frac{\kappa}{\kappa + 1} \doteq 0.331409$.
Consider three distinct sites $v, v', v'' \in V$ and
let $f = \Vor|_\Sigma v \cap \Vor|_\Sigma v' \cap \Vor|_\Sigma v''$.
Let $\ell_\tau \subset \R^3$ be the line containing
all points equidistant to the sites $v$, $v'$, and $v''$.

%%%UNCUT
% Then $f$ contains at most one point (i.e., a restricted Voronoi vertex).
% Moreover, if $f$~contains a point,
% $\Sigma$ is not tangent to $\ell_\tau$ at that point.
Then $f$ contains at most one point, and
$\Sigma$ is not tangent to $\ell_\tau$ at that point.
\end{corollary}

\begin{proof}
Suppose without loss of generality that $v$ is the vertex at
the largest plane angle of $\triangle vv'v''$.
As $V$ is an $\epsilon$-sample, for every point $y \in \Vor|_\Sigma v$,
$|vy| \leq \epsilon \, \lfs(y)$.
By the Feature Translation Lemma (Lemma~\ref{lem:ftl}),
$|vy| \leq \frac{\epsilon}{1 - \epsilon} \lfs(v) < \kappa \, \lfs(v)$.
The result follows by Lemma~\ref{lem:onlyonevertex}.
\end{proof}

\section{Restricted Voronoi edges are topological intervals}
\label{1ball}

Theorem~\ref{thm:homeocell}
%%% UNCUT
% , Lemma~\ref{lem:pathsincell},
and Corollary~\ref{cor:homeocelleps} give
conditions under which restricted Voronoi cells are
topological closed disks.
%%% UNCUT
% and no cell intersects the interior of another cell.
Lemma~\ref{lem:onlyonevertex} and Corollary~\ref{cor:onlyonevertexeps} give
conditions under which the intersection of
three distinct restricted Voronoi cells is at most one point.
% justifying the name ``restricted Voronoi vertex.''
What about an intersection of two distinct restricted Voronoi cells?
That could be empty, a restricted Voronoi vertex, or
a {\em restricted Voronoi edge}, the last being
a nonempty intersection of $\Sigma$ with a $2$-face of $\Vor \, V$.
Under suitable sampling conditions,
Lemma~\ref{lem:homeoedge}, below, guarantees closed ball property~B:
each restricted Voronoi edge is a topological interval.

The next four lemmas derive conditions in which
an intersection $\Vor|_\Sigma v \cap \Vor|_\Sigma w$ is
one interval or one isolated point.
Lemma~\ref{lem:trivertexnormalseps} applies to $0.32$-samples, whereas
Lemma~\ref{lem:trivertexnormals} applies to $0.41$-Voronoi samples, and
Lemma~\ref{lem:homeoedge} concludes both.
Lemma~\ref{lem:rvedge} has no sampling requirements,
only topological requirements.
% see the full-length paper for a proof~\cite{shewchuk26b}.

\begin{lemma}
\label{lem:rvedge}
Let $\Vor|_\Sigma V$ be a restricted Voronoi diagram.
Suppose that every restricted Voronoi cell is a topological closed disk and
every intersection of three distinct restricted Voronoi cells is
either empty or a lone point (henceforth called a restricted Voronoi vertex).
Suppose also that no restricted Voronoi cell intersects
the interior of another. % restricted Voronoi cell.

Then for every pair of distinct sites $v, w \in V$,
$\Vor|_\Sigma v \cap \Vor|_\Sigma w$ is one of the following:
\begin{itemize}
\item
empty,
\item
a topological circle containing no restricted Voronoi vertex, or
\item
a union of disjoint topological closed intervals and isolated points, where
each isolated point is a restricted Voronoi vertex and
each interval contains exactly two restricted Voronoi vertices
which are its endpoints.
\end{itemize}

Moreover, if each connected component of $\Sigma$ has
at least three sites in $V$ lying on it, then
the possibility that $\Vor|_\Sigma v \cap \Vor|_\Sigma w$ is
a topological circle is eliminated,
every restricted Voronoi cell has at least two restricted Voronoi vertices
% and at least two restricted Voronoi edges
on its boundary, and
%%%CUT
every connected component of $\Sigma$ has
at least two restricted Voronoi vertices on it.
\end{lemma}

\begin{proof}
As $\Sigma$ is a closed surface and
$\Sigma$ is also a union of restricted Voronoi cells,
which are topological closed disks,
each point on the boundary of each restricted Voronoi cell is shared with
at least one other restricted Voronoi cell.
By assumption, no interior point of a cell is shared with another cell;
cells intersect each other only on their boundaries.

Consider a site $v$, its restricted Voronoi cell $\Vor|_\Sigma v$, and
the cell's boundary $O$, which is a topological circle.
There are only finitely many restricted Voronoi vertices, as
there are only finitely many intersections of three restricted Voronoi cells.
If two or more restricted Voronoi vertices lie on $O$,
they subdivide $O$ into two or more topological closed intervals.
% (as restricted Voronoi vertices are isolated points).
Let $I$ be one of these topological intervals, or
let $I = O$ if $O$ contains fewer than two restricted Voronoi vertices.
The subset of $I$ obtained by removing its restricted Voronoi vertices is
path-connected, so the points in that subset are all shared with
one and only one other site~$w$; hence
$I \subseteq \Vor|_\Sigma v \cap \Vor|_\Sigma w$.
It follows that for every distinct $v, w \in V$,
$\Vor|_\Sigma v \cap \Vor|_\Sigma w$ is either a topological circle or
a union of intervals and restricted Voronoi vertices.

The intersection of two restricted Voronoi cells $C_1$ and $C_2$ cannot include
two intervals that are not disjoint---that is,
two intervals that share a restricted Voronoi vertex---because if they did,
the shared vertex would lie on the boundary of a third cell, which implies that
either $C_1$ or $C_2$ is not a topological closed disk,
% (as its boundary branches four or more ways at the vertex).
as Figure~\ref{voredges} (left) illustrates.
Likewise, if $C_1 \cap C_2$ is a topological circle,
no restricted Voronoi vertex can lie on the circle, because if one did,
either $C_1$ or $C_2$ would not be a topological closed disk,
as Figure~\ref{voredges} (center) illustrates.
This establishes the lemma's first claim.

\begin{figure}
\centerline{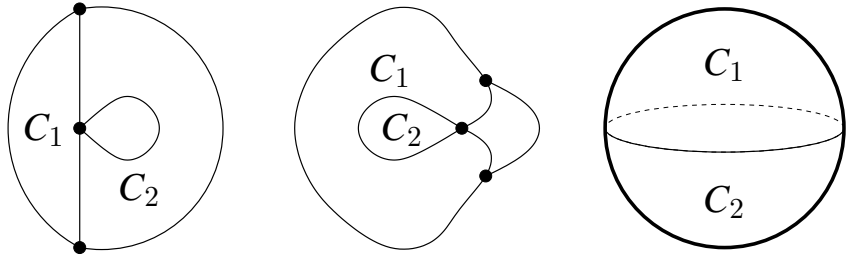}

\caption{\label{voredges}
Left:  Cells $C_1$ and $C_2$ share two edges that meet at a vertex;
either $C_1$ or $C_2$ is not a topological disk.
Center:  Cells $C_1$ and $C_2$ share a circle with one or more vertices on it;
either $C_1$ or $C_2$ is not a topological disk.
Right:  Cells $C_1$ and $C_2$ share a circle with no vertex on it;
if they are topological disks,
their union is topological sphere, a connected component of $\Sigma$.
}
\end{figure}

If the intersection of two cells is a topological circle
(with no restricted Voronoi vertex), then
as the two cells are topological disks, their union is a topological sphere,
as Figure~\ref{voredges} (right) illustrates.
This sphere is an entire connected component of $\Sigma$ with
no other cell on it.
Hence, if at least three sites in $V$ lie on
each connected component of $\Sigma$, then
no two cells have a topological circle as their intersection.
The lemma's final paragraph follows.
% Hence every cell has at least two restricted Voronoi vertices and
% at least two restricted Voronoi edges on its boundary, and
% every connected component of $\Sigma$ has
% at least two restricted Voronoi vertices on it.
\end{proof}

\begin{lemma}
\label{lem:trivertexnormalseps}
% Let $V$ be an $\epsilon$-sample of $\Sigma$ for some $\epsilon \leq 0.3245$.
Let $u$ be a restricted Voronoi vertex and
let $\tau = \triangle vv'v''$ be its dual restricted Delaunay triangle.
%%%CUT
(If $u$'s dual face is a polygon, you may choose any three of its vertices.)
% Let $n_u$ be an outside-facing vector normal to $\Sigma$ at $u$,
% let $n_v$ be an outside-facing vector normal to $\Sigma$ at $v$, and
Let $n_\tau$ be a vector normal to $\tau$,
directed so that $\angle (n_v, n_\tau) \leq 90^\circ$
(as illustrated in Figure~\ref{trivertex}).
% such that $n_u \cdot n_\tau \geq 0$.
% (Hence if we translate the vector $n_\tau$ to have its origin at $u$,
% it points toward the outside of $\Sigma$ or it is tangent to $\Sigma$.)
% Then $\angle (n_v, n_\tau) < 90^\circ$; equivalently, $n_v \cdot n_\tau > 0$.
Let $s = |vu| = |v'u| = |v''u|$ and suppose that $s \leq 0.3245 \, \lfs(u)$.

% Then the angles $\angle (n_u, n_\tau)$ and $\angle (n_v, n_\tau)$ are
% either both less than $90^\circ$ or both greater than $90^\circ$
% (depending on which way $n_\tau$ is directed).
% Equivalently, the dot products $n_u \cdot n_\tau$ and $n_v \cdot n_\tau$ are
% either both positive or both negative.
Then $\angle (n_u, n_\tau) < 90^\circ$ and $\angle (n_v, n_\tau) < 90^\circ$.
Equivalently, $n_u \cdot n_\tau$ and $n_v \cdot n_\tau$ are positive.
\end{lemma}

\begin{proof}
Let $w \in \{ v, v', v'' \}$ be the vertex at $\tau$'s largest plane angle.
As $|vu| = |wu| = s \leq 0.3245 \, \lfs(u)$,
by the Normal Variation Lemma (Lemma~\ref{lem:nvl}),
$\angle (n_v, n_u) \leq \eta(0.3245) < 19.21^\circ$ and
$\angle (n_w, n_u) \leq \eta(0.3245) < 19.21^\circ$, where $\eta(\delta) =
\arccos \left( 1 - \frac{\delta^2}{2 \sqrt{1 - \delta^2}} \right)$.
% and $\eta(0.3245) < 19.21^\circ$.

Let $r$ be $\tau$'s circumradius.
As Figure~\ref{trivertex} illustrates, $r \leq s$, because
$u$ lies on the line perpendicular to $\tau$ through $\tau$'s circumcenter
and $r$ is the distance from $v$ to that line.
By the Feature Translation Lemma (Lemma~\ref{lem:ftl}),
$\lfs(u) \leq \lfs(v) / (1 - 0.3245)$ and
$\lfs(u) \leq \lfs(w) / (1 - 0.3245)$.
Hence $r \leq 0.3245 \, \lfs(u) \leq \frac{0.3245}{0.6755} \, \lfs(v)$ and
$r \leq \frac{0.3245}{0.6755} \, \lfs(w)$.

\begin{figure}
\centerline{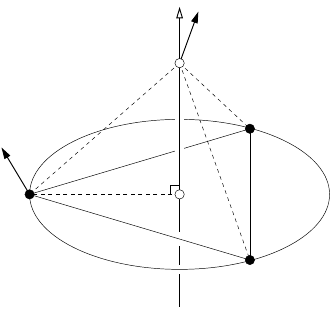}

\caption{\label{trivertex}
The sites $v$, $v'$, and $v''$ and the restricted Voronoi vertex $u$
lie on $\Sigma$ (not shown).
}
\end{figure}

If $\tau$'s plane angle at the vertex $v$ is $53.932^\circ$ or greater, then
by the Triangle Normal Lemma (Lemma~\ref{lem:tnl}),
$\sin \angle (n_v, n_\tau) \leq r \cot 26.966^\circ / \lfs(v)
% \leq \frac{0.3245}{0.6755} \, \cot 26.966^\circ
< (0.3245 / 0.6755) \cdot 1.9655 < 0.9442$.
Therefore, $\angle (n_v, n_\tau) < 70.77^\circ$ and
$\angle (n_u, n_\tau) \leq \angle (n_v, n_u) + \angle (n_v, n_\tau) <
19.21^\circ + 70.77^\circ = 89.98^\circ$.
% which is less than $90^\circ$ as claimed.

Otherwise, $\tau$'s plane angle at $v$ is less than $53.932^\circ$, so
$\tau$'s largest plane angle (at $w$) is
greater than $(180^\circ - 53.932^\circ) / 2 = 63.034^\circ$.
By the Triangle Normal Lemma, $\sin \angle (n_w, n_\tau) \leq
r \cot 31.517^\circ / \lfs(w)
% \leq \frac{0.3245}{0.6755} \, \cot 31.517^\circ
< (0.3245 / 0.6755) \cdot 1.6308 < 0.78342$.
Therefore, either $\angle (n_w, n_\tau) < 51.575^\circ$ or
$\angle (n_w, n_\tau) > 128.425^\circ$.
% depending on which way $n_\tau$ is directed.
% Suppose without loss of generality that $n_\tau$ is directed so that
% $\angle (n_w, n_\tau) < 51.575^\circ$.
% Then
The latter case
% ($\angle (n_w, n_\tau) > 128.425^\circ$)
is not possible, because
$\angle (n_w, n_\tau) \leq
\angle (n_w, n_u) + \angle (n_u, n_v) + \angle (n_v, n_\tau) <
19.21^\circ + 19.21^\circ + 90^\circ = 128.42^\circ$.
In the former case,
$\angle (n_u, n_\tau) \leq \angle (n_w, n_u) + \angle (n_w, n_\tau) <
19.21^\circ + 51.575^\circ = 70.785^\circ$ and
$\angle (n_v, n_\tau) \leq
\angle (n_v, n_u) + \angle (n_w, n_u) + \angle (n_w, n_\tau) <
19.21^\circ + 19.21^\circ + 51.575^\circ = 89.995^\circ$.
% confirming that both angles are less than $90^\circ$.
\end{proof}

\begin{lemma}
\label{lem:trivertexnormals}
Let $u$ be a restricted Voronoi vertex and
let $\tau = \triangle vv'v''$ be its dual restricted Delaunay triangle.
%%%CUT
(If $u$'s dual face is a polygon, you may choose any three of its vertices.)
Let $w \in \{ v, v', v'' \}$ be the vertex at $\tau$'s largest plane angle.
% Let $n_v$ be an outside-facing vector normal to $\Sigma$ at $v$,
% let $n_u$ be an outside-facing vector normal to $\Sigma$ at $u$, and
Let $n_\tau$ be a vector normal to~$\tau$,
directed so that $\angle (n_v, n_\tau) \leq 90^\circ$.
Let $s = |vu| = |v'u| = |v''u| = |wu|$ and suppose that
% $|vu| \leq 0.419542 \, \lfs(v)$ and $|wu| \leq 0.419542 \, \lfs(w)$.
$s \leq 0.4132 \, \lfs(v)$ and $s \leq 0.4132 \, \lfs(w)$.

% Then the angles $\angle (n_u, n_\tau)$ and $\angle (n_v, n_\tau)$ are
% either both less than $90^\circ$ or both greater than $90^\circ$
% (depending on which way $n_\tau$ is directed).
% Equivalently, the dot products $n_u \cdot n_\tau$ and $n_v \cdot n_\tau$ are
% either both positive or both negative.
Then $\angle (n_u, n_\tau) < 90^\circ$ and $\angle (n_v, n_\tau) < 90^\circ$.
Equivalently, $n_u \cdot n_\tau$ and $n_v \cdot n_\tau$ are positive.
\end{lemma}

% The proof of Lemma~\ref{lem:trivertexnormals} is much like
% that of Lemma~\ref{lem:trivertexnormalseps};
% see the full-length paper~\cite{shewchuk26b}.

\begin{proof}
As $|vu| \leq 0.4132 \, \lfs(v)$ and $|wu| \leq 0.4132 \, \lfs(w)$,
by the Normal Variation Lemma (Lemma~\ref{lem:nvl}),
$\angle (n_v, n_u) \leq \eta(0.4132) < 25.008^\circ$ and
$\angle (n_w, n_u) \leq \eta(0.4132) < 25.008^\circ$, where $\eta(\delta) =
\arccos \left( 1 - \frac{\delta^2}{2 \sqrt{1 - \delta^2}} \right)$.
% and $\eta(0.4132) < 25.008^\circ$.

Let $r$ be $\tau$'s circumradius.
As Figure~\ref{trivertex} illustrates, $r \leq s$, because
$u$ lies on the line perpendicular to $\tau$ through $\tau$'s circumcenter
and $r$ is the distance from $v$ to that line.
So $r \leq 0.4132 \, \lfs(v)$ and $r \leq 0.4132 \, \lfs(w)$.

If $\tau$'s plane angle at the vertex $v$ is $49.023^\circ$ or greater, then
by the Triangle Normal Lemma (Lemma~\ref{lem:tnl}),
$\sin \angle (n_v, n_\tau) \leq r \cot 24.5115^\circ / \lfs(v) <
0.4132 \cdot 2.1932 < 0.906231$.
Therefore, $\angle (n_v, n_\tau) < 64.99^\circ$ and
$\angle (n_u, n_\tau) \leq \angle (n_v, n_u) + \angle (n_v, n_\tau) <
25.008^\circ + 64.99^\circ = 89.998^\circ$.
% which is less than $90^\circ$ as claimed.

Otherwise, $\tau$'s plane angle at $v$ is less than $49.023^\circ$, so
$\tau$'s largest plane angle (at $w$) is
greater than $(180^\circ - 49.023^\circ) / 2 = 65.4885^\circ$.
By the Triangle Normal Lemma, $\sin \angle (n_w, n_\tau) \leq
r \cot 32.74425^\circ / \lfs(w) < 0.4132 \cdot 1.5551 < 0.642568$.
Therefore, either $\angle (n_w, n_\tau) < 39.9836^\circ$ or
$\angle (n_w, n_\tau) > 140.0164^\circ$.
% depending on which way $n_\tau$ is directed.
% Suppose without loss of generality that $n_\tau$ is directed so that
% $\angle (n_w, n_\tau) < 39.9836^\circ$.
The latter case is not possible, because
$\angle (n_w, n_\tau) \leq
\angle (n_w, n_u) + \angle (n_v, n_u) + \angle (n_v, n_\tau) <
25.008^\circ + 25.008^\circ + 90^\circ = 140.016^\circ$.
In the former case,
$\angle (n_u, n_\tau) \leq \angle (n_w, n_u) + \angle (n_w, n_\tau) <
25.008^\circ + 39.9836^\circ = 64.9916^\circ$ and
$\angle (n_v, n_\tau) \leq
\angle (n_v, n_u) + \angle (n_w, n_u) + \angle (n_w, n_\tau) <
25.008^\circ + 25.008^\circ + 39.9836^\circ = 89.9996^\circ$.
% confirming that both angles are less than $90^\circ$.
\end{proof}

The final lemma shows that a nonempty intersection of
two restricted Voronoi cells is either a lone restricted Voronoi vertex
(a ``degenerate'' case where
four or more restricted Voronoi cells share a restricted Voronoi vertex,
whereas the common case is three cells sharing a vertex) or
% as opposed to the common case where
% exactly three restricted Voronoi cells share a restricted Voronoi vertex) or
a~lone topological interval justifying the name ``restricted Voronoi edge.''
%%% JRS:  UNCUT
% This is our favorite lemma after Lemma~\ref{lem:abovebelow};
% we are rather pleased with the argument.

\begin{lemma}
\label{lem:homeoedge}
Let $\Vor|_\Sigma V$ be a restricted Voronoi diagram that satisfies
all the \mbox{conditions} of Lemma~\ref{lem:rvedge},
including the condition that
at least three sites in $V$ lie on each connected component of $\Sigma$.
% Let $v \in V$ be a site.
% Suppose that $v$'s restricted Voronoi cell $\Vor|_\Sigma v$ is
% homeomorphic to a closed disk, and that
% for every site $w \in V$ such that
% $\Vor|_\Sigma v$ intersects $\cap \Vor|_\Sigma w$,
% $\Vor|_\Sigma w$ is homeomorphic to a closed disk.
% Suppose also that for every pair of sites $w, w' \in V$,
% $\Vor|_\Sigma v \cap \Vor|_\Sigma w \cap \Vor|_\Sigma w'$ is
% either a lone point (i.e., a restricted Voronoi vertex) or empty.
% Suppose also that at least two distinct restricted Voronoi vertices lie on
% the boundary of $\Vor|_\Sigma v$.
% Suppose that one of the following is true:  either
% % $V$ is an $\epsilon$-sample of $\Sigma$ for some $\epsilon \leq 0.3245$, or
% for every site $v \in V$ and
% every restricted Voronoi vertex $u \in \Vor|_\Sigma v$,
% $|vu| \leq 0.3245 \, \lfs(u)$; or
% for every site $v \in V$ and
% every restricted Voronoi vertex $u \in \Vor|_\Sigma v$,
% $|vu| \leq 0.4132 \, \lfs(v)$.
Suppose that
%%%CUT
closed ball property~D and generic intersection properties~E and~F hold:
% properties~D, E, and~F hold:
no vertex of $\Vor \, V$ lies on $\Sigma$, and
no edge and no $2$-face of $\Vor \, V$ intersects $\Sigma$ tangentially.
Suppose also that for every restricted Voronoi vertex $u \in \Vor|_\Sigma V$,
either
\begin{itemize}
\item
\mbox{$|vu| \leq 0.3245 \, \lfs(u)$}
for every site $v$ such that $u \in \Vor|_\Sigma v$, or
\item
$|vu| \leq 0.4132 \, \lfs(v)$
for every site $v$ such that $u \in \Vor|_\Sigma v$.
\end{itemize}
Let $v, w \in V$ be two distinct sites,
let $F = \Vor \, v \cap \Vor \, w$, and
let $f = \Vor|_\Sigma v \cap \Vor|_\Sigma w = F \cap \Sigma$.

Then one of these three claims holds:
\begin{itemize}
\item
$f$ is empty;
\item
$f$~is a lone point (a restricted Voronoi vertex) and $F$ is an edge; or
\item
$f$~is homeomorphic to a closed interval and $F$ is a $2$-face.
\end{itemize}
\end{lemma}

\begin{proof}
Suppose $f$ is nonempty.
By Lemma~\ref{lem:rvedge}, $f$ is
a union of disjoint topological $1$-balls and isolated points.
As $f$ is nonempty and no vertex of $\Vor \, V$ lies on $\Sigma$
(closed ball property~D), $F$ is not a vertex of $\Vor \, V$.
Hence $F$ is either an edge or a $2$-face of $\Vor \, V$.
If $F$ is a Voronoi edge, then
$F$ is the intersection of three or more Voronoi cells and thus
$f$ is the intersection of three or more restricted Voronoi cells; so
% Lemma~\ref{lem:onlyonevertexeps} or~\ref{lem:onlyonevertex} implies that
by assumption (the conditions of Lemma~\ref{lem:rvedge}),
$f$~is a lone restricted Voronoi vertex and the result holds.

Only the case where $F$ is a $2$-face of $\Vor \, V$ remains.
% By assumption, $F$ does not intersect $\Sigma$ tangentially,
% % Lemma~\ref{lem:onlyonevertexeps} or~\ref{lem:onlyonevertex} imply that
% no edge of $F$ intersects $\Sigma$ tangentially, and
% no vertex of $F$ intersects $\Sigma$ at all.
As closed ball property~D and generic intersection properties~E and~F hold,
% As properties~D, E, and F hold,
$\Sigma \cap F$ cannot contain any isolated points;
so $f$ is a~union of disjoint topological intervals.
By Lemma~\ref{lem:rvedge},
each interval contains exactly two restricted Voronoi vertices, its endpoints.
Both of them lie on the boundary of $F$.

Let $n_v$ be an outside-facing vector normal to $\Sigma$ at $v$.
We will see shortly that $F$ is not perpendicular to $n_v$.
Assign a direction to each edge of $F$ such that
the angle between $n_v$ and each directed edge is at most $90^\circ$,
as illustrated in Figure~\ref{chains}.
As $F$ is a convex polygon,
we thus partition its edges into two chains,
each monotone in the direction~$n_v$.
(An edge perpendicular to $n_v$---there are at most two---can
be assigned to either chain.)

\begin{figure}
\centerline{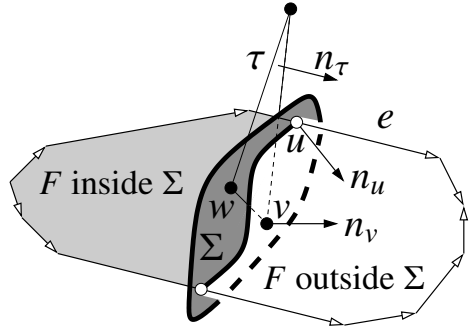}

\caption{\label{chains}
A Voronoi $2$-face $F$, with its edges partitioned into two chains
monotone in $n_v$.
}
\end{figure}

Consider a restricted Voronoi vertex $u$ on an edge $e$ of $F$ and
% let $n_u$ be an outside-facing vector normal to $\Sigma$ at~$u$, and
let $\tau$ be $u$'s dual restricted Delaunay triangle (or polygon).
Let $n_\tau$ be a vector that points in the same direction as~$e$, and
observe that $n_\tau$ is normal to $\tau$ and
$\angle (n_v, n_\tau) \leq 90^\circ$.
By Lemma~\ref{lem:trivertexnormalseps} or Lemma~\ref{lem:trivertexnormals},
$\angle (n_v, n_\tau) < 90^\circ$ and
$\angle (n_u, n_\tau) < 90^\circ$.
The former implies that $F$ is not perpendicular to $n_v$ (as promised).
The latter implies that as one walks along one of the directed chains
(monotone with respect to~$n_v$),
one might encounter a restricted Voronoi vertex where the chain passes
from inside $\Sigma$ to outside $\Sigma$, but not from outside to inside.
Therefore, there can be only one restricted Voronoi vertex on each chain, and
only two restricted Voronoi vertices on $F$.
It follows that $\Vor|_\Sigma v \cap \Vor|_\Sigma w$ is just
a single topological interval.
\end{proof}

\section{The restricted Delaunay triangulation is homeomorphic to
         ${\bf \Sigma}$}

I conclude with homeomorphism theorems for
$0.3245$-samples and $0.4132$-Voronoi samples.

% \begin{theorem}[Topological Ball Theorem \cite{edelsbrunner97}]
% \label{TBT}
% If $\Sigma$ and $V$ satisfy the closed ball property and
% the generic intersection property, then
% the underlying space of $\Del|_\Sigma S$ is
% homeomorphic\index{homeomorphism} to $\Sigma$.
% \end{theorem}

\begin{theorem}
\label{thm:homeosample1}
Let $V$ be a finite $\epsilon$-sample of $\Sigma$
for some $\epsilon \leq 0.3245$.
Suppose that no vertex of the three-dimensional Voronoi diagram $\Vor \, V$
lies on $\Sigma$.
Then the underlying space of the restricted Delaunay triangulation,
$|\Del|_\Sigma V|$, is homeomorphic to $\Sigma$.
\end{theorem}

\begin{proof}
Corollary~\ref{cor:homeocelleps} guarantees closed ball property~A and
% every intersection of $\Sigma$ with a Voronoi cell in $\Vor \, V$
% (i.e., every restricted Voronoi cell in $\Vor|_\Sigma V$)
% is homeomorphic to a closed disk.
generic intersection property~E.
% $\Sigma$~does not intersect any Voronoi $2$-face tangentially.
Corollary~\ref{cor:onlyonevertexeps} guarantees closed ball property~C and
% Lemma~\ref{lem:homeoedge} guarantees closed ball properties~B and~C and
% every nonempty intersection of $\Sigma$ with an edge of $\Vor \, V$
% (i.e., every nonempty intersection of three or more restricted Voronoi cells)
% contains exactly one point (i.e., a restricted Voronoi vertex).
generic intersection property~F.
% $\Sigma$ does not intersect any Voronoi edge tangentially.
Closed ball property~D holds by assumption.
By Corollary~\ref{cor:homeocelleps}, no restricted Voronoi cell intersects
the interior of another restricted Voronoi cell and
every connected component of $\Sigma$ has at least six sites on it;
all the preconditions of Lemmas~\ref{lem:rvedge} and~\ref{lem:homeoedge} are
satisfied.
Lemma~\ref{lem:homeoedge} guarantees closed ball property~B.
% By Lemma~\ref{lem:homeoedge},
% every intersection of two restricted Voronoi cells is either
% a topological interval (i.e., a restricted Voronoi edge),
% a restricted Voronoi vertex, or empty.
% Therefore, every nonempty intersection of $\Sigma$ with
% a $2$-face $F \in \Vor \, V$ is a topological interval,
% proving closed ball property~B.
% (The intersection cannot be an isolated point because
% neither $F$ nor any edge of $F$ intersects $\Sigma$ tangentially.)
% no vertex of $\Vor \, V$ lies on $\Sigma$.
As $\Sigma$ and $V$ satisfy
the preconditions A--F of the Topological Ball Theorem~\cite{edelsbrunner97},
$|\Del|_\Sigma V|$ is homeomorphic to~$\Sigma$.
\end{proof}

\begin{theorem}
\label{thm:homeosigma1}
% Let $V$ be a nonempty, finite set of points on $\Sigma$.
Let $V$ be a finite $\epsilon$-Voronoi sample of $\Sigma$
for some $\epsilon \leq 0.4132$.
Suppose that no vertex of the three-dimensional Voronoi diagram $\Vor \, V$
lies on $\Sigma$.
% Suppose that for every site $v \in V$ and every point $x \in \Vor|_\Sigma v$,
% $|vx| < 0.4132 \, \lfs(v)$.
% $|vx| < \xi \, \lfs(v)$, where
% $\xi = \sqrt{(\sqrt{5} - 1) / 2} \doteq 0.786151$.
% Suppose moreover that for every site $v \in V$ and
% every point $u \in \Vor|_\Sigma v$ that lies in the intersection of
% three or more restricted Voronoi cells, $|vu| < 0.4132 \, \lfs(v)$.
Then the underlying space of the restricted Delaunay triangulation,
$|\Del|_\Sigma V|$, is homeomorphic to $\Sigma$.
\end{theorem}

\begin{proof}
Identical to the proof of Theorem~\ref{thm:homeosample1}, except that
Corollary~\ref{cor:homeocelleps} is replaced by Theorem~\ref{thm:homeocell} and
Lemmas~\ref{lem:raybisector}, \ref{lem:sixsites}, and~\ref{lem:pathsincell}; and
Corollary~\ref{cor:onlyonevertexeps} is replaced by
Lemma~\ref{lem:onlyonevertex}.
% Corollary~\ref{cor:homeoedge} is replaced by
% Lemma~\ref{lem:homeoedge}.
\end{proof}

\section{Restricted Voronoi cells in higher dimensions}
\label{higherd}

This section shows how to generalize the claim that
restricted Voronoi cells are topological balls with star-shaped projections
to smooth, closed manifolds of arbitrary dimension embedded
in arbitrary-dimensional spaces.
The sampling constants do not deteriorate.
% , except that our sampling requirement uses a strict inequality
% ($|vx| < \xi \, \lfs(v)$ instead of $|vx| \leq \xi \, \lfs(v)$).

Consider the restricted Voronoi diagram of
a smooth, closed $k$-manifold $\Sigma$ embedded
in the Euclidean space $\R^d$.
The {\em codimension} is $d - k$.
Every point $p \in \Sigma$ has a $k$-dimensional tangent space $T_p\Sigma$ and
a $(d - k)$-dimensional normal space $N_p\Sigma$, both passing through $p$.
Both $T_p\Sigma$ and $N_p\Sigma$ are {\em flats},
also known as {\em affine subspaces}.

Lemma~\ref{lem:abovebelow} needs the biggest changes to generalize
to higher codimensions; see
Lemmas~\ref{lem:pierce} and~\ref{lem:abovebelowhi} below.
Lemmas~\ref{lem:raybisectorhi} and~\ref{lem:neighborinjecthi} below generalize
Lemmas~\ref{lem:raybisector} and~\ref{lem:neighborinject}.
The remaining results in Section~\ref{2ball} generalize easily, with
no need to restate the proofs, although for some of them
a few word substitutions are needed.
In the proof of Lemma~\ref{lem:pathincell},
the plane $\Lambda$ becomes a hyperplane and
the line $N_v\Sigma$ becomes a normal space of dimension $d - k$.
In Theorem~\ref{thm:homeocell} and its proof,
replace the word ``disk'' with ``$k$-ball.''
In the proof of Lemma~\ref{lem:slicecontinuous},
the arcs $a^+$ and $a^-$ become spherical caps.
The proofs of Lemmas~\ref{lem:pathsincell}, \ref{lem:sixsites},
and~\ref{lem:ftl} require no changes.
At the end of this section,
Theorem~\ref{thm:homeocellhi} generalizes Theorem~\ref{thm:homeocell} and
Corollary~\ref{cor:homeocellepshi} generalizes Corollary~\ref{cor:homeocelleps}
(but we do not restate the proofs).

An {\em lfs-ball} at a point $p \in \Sigma$ is
an open $d$-ball with radius $\lfs(p)$ that is tangent to $\Sigma$ at $p$.
Every lfs-ball is a subset of an open medial ball, so
every lfs-ball is disjoint from $\Sigma$.
% (Note that in~$\R^d$,
% all medial balls and lfs-balls are $d$-balls by definition.)
Lfs-balls are the counterparts of $B$ and $B'$ in Figure~\ref{oopposite}.
If the codimension is $1$, there are two lfs-balls at $p$---one inside $\Sigma$
and one outside $\Sigma$.
If the codimension is $2$ or greater, there are infinitely many lfs-balls.

Let $\mathbb{B}$ be the (open) union of all the lfs-balls.
In codimension $1$, $\mathbb{B}$ is just two disjoint open balls
(as in Figure~\ref{oopposite}), but in codimension $2$ or greater,
the boundary of $\mathbb{B}$ is a torus with inner radius zero (a horn torus).
Starting from Figure~\ref{oopposite},
imagine the solid of revolution of $B$ around $T_p \Sigma$.
We call $\mathbb{B}$ (the interior, not its boundary)
the {\em solid torus} for $p$.
Topologically, $\mathbb{B}$ is
the product of a $(d - k - 1)$-sphere and an open $(k + 1)$-ball.
% Geometrically, $\mathbb{B}$ is
% the Minkowski sum of $\dot{\mathbb{B}}$ and an open $d$-ball.

Let $Q$ be the open $d$-ball with center $p$ and radius $\lfs(p)$.
The {\em torus disk} for $p$ is $\mathbb{D} = N_p\Sigma \cap Q =
\{ z \in N_p\Sigma : |pz| < \lfs(p) \}$,
the open $(d - k)$-ball on $N_p\Sigma$ with center $p$ and radius $\lfs(p)$.
(Here, ``disk'' does not mean it's $2$-dimensional;
it means it's less than $d$-dimensional.)

Let $\bar{\mathbb{D}}$ denote the closure of $\mathbb{D}$.
Let $\partial \mathbb{D} = \bar{\mathbb{D}} \setminus \mathbb{D}$.
Thus, $\partial \mathbb{D} = \{ z \in N_p\Sigma : |pz| = \lfs(p) \}$ is
the $(d - k - 1)$-sphere on $N_p\Sigma$ with center $p$ and radius $\lfs(p)$.
Observe that $\partial \mathbb{D}$ is
the set of the centers of all the lfs-balls.
Call $\partial \mathbb{D}$ the {\em torus skeleton} for $p$, because
it happens to be the inside component of the medial axis of
the horn torus bounding $\mathbb{B}$.
If the codimension is $1$, the torus skeleton consists of
just two points (the centers of the two lfs-balls).
Otherwise, the torus skeleton is
a more traditional circle or sphere (of dimension $d - k - 1$).

For every point $p \in \Sigma$, the significance of $p$'s solid torus is that
it is disjoint from $\Sigma$.
Given a site $v$,
the significance of $v$'s torus disk is that it is in the interior of
the $d$-dimensional Voronoi cell $\Vor \, v$.

% Lemma~\ref{lem:nvl} is a Normal Variation Lemma for
% manifolds of codimension~$1$.
% If the codimension is greater than $1$,
% we must use another Normal Variation Lemma with a weaker bound.

% \begin{lemma}[Normal Variation Lemma for Codimension $2$ and
%               Higher~\cite{khoury19}]
% \label{ptnormal2}
% Let $\Sigma \subset \mathbb{R}^d$ be
% a smooth, closed $k$-manifold for any $k < d$.
% Consider two points $p, q \in \Sigma$ and let $\delta = |pq| / \lfs(p)$.
% % Let $N_p\Sigma$ and $N_q\Sigma$ be the $(d - k)$-flats normal to $\Sigma$
% % at $p$ and $q$, respectively.
% If $\delta < \sqrt{\left( \sqrt{5} - 1 \right) / 2} \doteq 0.786151$, then
% $\angle (N_p\Sigma, N_q\Sigma) = \angle (T_p\Sigma, T_q\Sigma) \leq
% \eta_2(\delta)$ where
% \begin{equation}
% \eta_2(\delta) = \arccos
%   \sqrt{1 - \frac{\delta^2}{\sqrt{1 - \delta^2}}}
% \approx
% \delta + \frac{5}{12} \delta^3 + O(\delta^5).
% \label{codim2bound}
% \end{equation}
% \end{lemma}

\begin{lemma}
\label{lem:pierce}
Let $\Sigma \subset \R^d$ be a smooth, closed $k$-manifold.
Let $v \in \Sigma$ be a point and let $\mathbb{D}$ be its torus disk.
Let $z \in \R^d$ be any point.
Then there is a point $o$ on the torus skeleton $\partial \mathbb{D}$ such that
the line segment $oz$ intersects $\Sigma$.
\end{lemma}

\begin{proof}
If $z \in N_v\Sigma$, then
let $o$ be the point on $\partial \mathbb{D}$ that is farthest from $z$.
The line segment $oz$ contains $v$ and the result follows.

Otherwise, let $\Psi$ be the affine hull of $N_v\Sigma \cup \{ z \}$;
$\Psi$ is a $(d - k + 1)$-flat.
%By using $\Psi$ to take a cross section of the ambient space $\R^d$,
%we eliminate all but one dimension of $v$'s tangent space;
%the surviving dimension accommodates $z$.
%The cross section of $v$'s solid torus, $\mathbb{B} \cap \Psi$, is
%a solid torus in $d - k + 1$ dimensions that has ``room'' only
%for a $1$-manifold to pass through its hole at $v$.
%In a sufficiently small neighborhood $N \subset \Sigma$ of $v$,
%the cross section $N \cap \Psi$ is a topological interval.
Consider first the case where $\Sigma$ and $\Psi$ intersect transversally.
Then $\Sigma \cap \Psi$ is a smooth, closed $1$-manifold.
Let $L$ (for ``loop'') be the connected component of $\Sigma \cap \Psi$ that
contains $v$.

Let $\Upsilon$ be the union of all the line segments that connect
$z$ to a point on $\partial \mathbb{D}$.
Thus $\Upsilon \cup \mathbb{D}$ is the boundary of
a cone with base $\bar{\mathbb{D}}$ and tip $z$.
Observe that this cone (the convex hull of $\Upsilon \cup \mathbb{D}$) is
a subset of $\Psi$ that has the same dimension as $\Psi$, so
the cone's boundary $\Upsilon \cup \mathbb{D}$ cuts
$\Psi$ into an inside piece and an outside piece.

At $v$, the loop~$L$ is perpendicular to $\mathbb{D}$
(because $T_vL \subseteq T_v\Sigma$ and $\mathbb{D} \subset N_v\Sigma$).
As $L \subset \Psi$ and $L$ is smooth,
$L$ passes from outside $\Upsilon \cup \mathbb{D}$ to inside at $v$.
As $\Sigma$~does not intersect any point on $\mathbb{D} \setminus \{ v \}$,
the loop $L$ must intersect $\Upsilon$ at some point~$q \in \Upsilon$.
Thus there is a point $o \in \partial \mathbb{D}$ such that $q \in oz$, and
the result follows.

Next, consider the case where
$\Sigma$ and $\Psi$ do not intersect transversally.
Suppose for the sake of contradiction that
there is no point $o \in \partial \mathbb{D}$ such that
$oz$ intersects~$\Sigma$; thus $\Sigma \cap \Upsilon = \emptyset$.
Let $\beta = \min \{ |pq| : p \in \Sigma$ and $q \in \Upsilon \}$.
As both $\Sigma$ and $\Upsilon$ are compact, $\beta$ is defined and positive.

By the Thom Transversality Theorem,
there is a $(d - k + 1)$-manifold $\Psi' \subset \R^d$ that is
an arbitrarily small perturbation of $\Psi$ (but not necessarily flat) such that
$\Sigma$ and $\Psi'$ intersect transversally.
More precisely, for every real $\delta > 0$,
there exists a diffeomorphism $\mu : \Psi \rightarrow \Psi'$
such that for every point $p \in \Psi$, $|p\mu(p)| < \delta$ and
$\Sigma$ and $\Psi'$ intersect transversally in $\R^d$.
Then $\Sigma \cap \Psi'$ is a smooth, closed $1$-manifold.
Moreover, as $\Psi$ intersects $\Sigma$ transversally at $v$ and
at all points sufficiently close to $v$ (because $\Sigma$ is smooth), and
$\Sigma$ does not intersect $\mathbb{D}$ anywhere except $v$,
we can choose $\mu$ so that $\mu(q) = q$ for all points $q \in \mathbb{D}$.
Thus $\mu(v) = v$, $\mu(\mathbb{D}) = \mathbb{D}$, and $v \in \Psi'$.
Consider such a diffeomorphism with $\delta < \beta$.
Let $L$ be the connected component of $\Sigma \cap \Psi'$ that contains $v$.

As $L \subset \Psi'$, $L$ is smooth, and
$T_vL \not\subseteq N_v\Sigma$ at $v$,
$L$ passes from inside $\mu(\Upsilon \cup \mathbb{D})$ to outside at $v$.
As $\Sigma$~does not intersect any point on $\mathbb{D} \setminus \{ v \}$,
the loop $L$ must intersect $\mu(\Upsilon)$ at some point~$q \in \mu(\Upsilon)$.
By assumption, $|q\mu^{-1}(q)| < \delta < \beta$.
But $q \in \Sigma$ and $\mu^{-1}(q) \in \Upsilon$, contradicting the fact that
$\beta = \min \{ |pq| : p \in \Sigma$ and $q \in \Upsilon \}$.
\end{proof}

\begin{lemma}
\label{lem:abovebelowhi}
Let $\Sigma \subset \R^d$ be a smooth, closed $k$-manifold.
Consider two points $v, x \in \Sigma$ such that $|vx| < \xi \, \lfs(v)$,
where $\xi = \sqrt{(\sqrt{5} - 1) / 2} \doteq 0.786151$.
Let $\mathbb{D}$ be the torus disk for $v$.
Then $T_x\Sigma \cap \mathbb{D}$ contains exactly one point.
\end{lemma}

\begin{proof}
% As $|vx| < \xi \, \lfs(v)$,
% by the Normal Variation Lemma (Lemma~\ref{ptnormal2}),
% $\angle (N_v\Sigma, N_x\Sigma) < 90^\circ$; so
% $T_x\Sigma$ does not intersect $N_v\Sigma$ at more than one point.
% As $\mathbb{D} \subset N_v\Sigma$,
% $T_x\Sigma \cap \mathbb{D}$ contains at most one point.
%
Suppose for the sake of contradiction that
$T_x\Sigma \cap \mathbb{D}$ does not contain exactly one point.
Then $x \neq v$, as $T_v\Sigma \cap \mathbb{D} = \{ v \}$.
There are two possibilities:  either $T_x\Sigma \cap \mathbb{D} = \emptyset$, or
$T_x\Sigma \cap \mathbb{D}$ contains more than one point.
% Then either $T_x\Sigma$ does not intersect $\mathbb{D}$ or
% $T_x\Sigma \cap \mathbb{D}$ contains multiple points; in the latter case,
% $\angle (T_v\Sigma, T_x\Sigma) = 90^\circ$.

We claim that there exists a closed halfspace $\Pi_v$ such that
$\mathbb{D} \subset \Pi_v$ and
the boundary $\Pi$ of $\Pi_v$ (a hyperplane) includes $T_x\Sigma$.
In the case where $T_x\Sigma \cap \mathbb{D} = \emptyset$,
there exists a hyperplane $\Pi \supseteq T_x\Sigma$ that
does not intersect $\mathbb{D}$,
because $\mathbb{D}$ is convex and $T_x\Sigma$ is a flat; and
we can choose $\Pi_v$ to be
the halfspace bounded by $\Pi$ that includes $\mathbb{D}$.
In case where $T_x\Sigma \cap \mathbb{D}$ contains more than one point,
$T_x\Sigma \cap N_v\Sigma$ contains more than one point
(as $\mathbb{D} \subset N_v\Sigma$) and hence is
a flat of dimension~$1$ or greater.
Therefore, there exists a hyperplane $\Pi$ that includes both
$T_x\Sigma$ and $N_v\Sigma$.
Then $\mathbb{D} \subset N_v\Sigma \subseteq \Pi$ and
we can choose $\Pi_v$ to be either closed halfspace bounded by $\Pi$.
In both cases, $T_x\Sigma \subseteq \Pi$, $v \in \mathbb{D} \subset \Pi_v$, and
$\partial \mathbb{D} \subset \Pi_v$.
(If the codimension is~$1$, then
$\Pi = T_x\Sigma$ and $\mathbb{D}$ is an open line segment.)

Among the open medial balls tangent to $\Sigma$ at $x$,
there are two tangent to $\Pi$ at~$x$; among those two,
let $B_m$ be the one with center $m \in \Pi_v$.
We claim that $B_m$ is bounded; it cannot be an open halfspace.
If $B_m$ is an open halfspace then its boundary is~$\Pi$ and
its closure is~$\Pi_v$; and
as $B_m \cap \Sigma = \emptyset$ and $v \in \Sigma \cap \Pi_v$, we have
$v \in \Pi$ and $T_v\Sigma \subseteq \Pi$.
Moreover, the facts that $v$~is the center of the disk $\mathbb{D}$,
$\mathbb{D} \subset \Pi_v$, and $v \in \Pi$ imply that
$\mathbb{D} \subset \Pi$ and thus $N_v\Sigma \subseteq \Pi$.
But it cannot be true that both $T_v\Sigma \subseteq \Pi$ and
$N_v\Sigma \subseteq \Pi$ (because that implies that $\Pi = \R^d$), so
$B_m$~is bounded as claimed.
Thus $B_m$ has a center $m$.
Observe that $xm$ is perpendicular to $\Pi$ and
$m$ lies on the medial axis of $\Sigma$.
% As $v \in \Pi_v$ but $v \not\in \Pi$, $\angle vxm < 90^\circ$.
% As $v \in \mathbb{D} \subset \Pi_v$, $\angle vxm \leq 90^\circ$.
% (Note that although some medial balls can be open halfspaces,
% $B_m$ cannot because $\angle vxm < 90^\circ$ and $v \not\in B_m$.)

By Lemma~\ref{lem:pierce},
there is a point $o$ on the torus skeleton $\partial \mathbb{D}$ such that
the line segment $om$ intersects $\Sigma$.
Let $B$ be the lfs-ball at $v$ with center $o$.
Then $B$ is disjoint from $B_m$, as neither $B$ nor $B_m$ intersects $\Sigma$.
%Intuitively, this is true because if $B_m$ intersected {\em every} lfs-ball,
%then it would block the passage of $\Sigma$ through
%the hole in the solid torus $\mathbb{B}$,
%making it impossible for a closed $k$-manifold to pass through
%$\mathbb{B}$ at $v$ and rejoin itself.
%See the footnote for an alternative (but longer) explanation.\footnote{
%}

As $v, m \in \Pi_v$, $o \in \partial \mathbb{D} \subset \Pi_v$, $x \in \Pi$, and
$xm$ is orthogonal to $\Pi$,
it follows that $\angle vxm \leq 90^\circ$ and $\angle oxm \leq 90^\circ$.
By Pythagoras' Theorem, $|ox|^2 + |mx|^2 \geq |om|^2$ and
$|vx|^2 + |mx|^2 \geq |vm|^2$.
As $m$ lies on the medial axis, $|vm| \geq \lfs(v)$
(by the definition of $\lfs$) and thus $|vx|^2 + |mx|^2 \geq \lfs(v)^2$.

The radii of $B$ and $B_m$ are $\lfs(v)$ and $|mx|$ respectively.
As $B$ and $B_m$ are disjoint, $|om| \geq \lfs(v) + |mx|$.
Combining this with the inequality $|ox|^2 + |mx|^2 \geq |om|^2$ gives
% $|ox|^2 + |mx|^2 \geq $(|mx| + \lfs(v))^2$
$|ox|^2 \geq \lfs(v)^2 + 2 \, \lfs(v) \, |mx|$.
Combining this with the inequality $|mx|^2 \geq \lfs(v)^2 - |vx|^2$ gives
$|ox|^2 \geq \lfs(v)^2 + 2 \, \lfs(v) \sqrt{\lfs(v)^2 - |vx|^2}$.

Let $o' = 2v - o$ be the point such that $v$ is the midpoint of $oo'$.
Then $o'$ lies on~$\partial \mathbb{D}$, opposite~$o$.
Let $B'$ be the lfs-ball touching $v$ with center $o'$.
Let ${\cal L}_v$ be the line through $o$, $v$, and $o'$.
Create a coordinate system with $v = (0, 0, \ldots, 0)$ and
$x = (x_h, x_v, 0, \ldots, 0)$ such that
$x_h$~is the distance from $x$ to ${\cal L}_v$
(the horizontal axis in Figure~\ref{oopposite}) and
$x_v$ is the coordinate of~$x$ in the direction parallel to ${\cal L}_v$
(the vertical axis in Figure~\ref{oopposite}).
Then
\[
|ox|^2 + |o'x|^2 =
x_h^2 + (x_v - \lfs(v))^2 + x_h^2 + (x_v + \lfs(v))^2 =
2 x_h^2 + 2 x_v^2 + 2 \, \lfs(v)^2 = 2 \, |vx|^2 + 2 \, \lfs(v)^2.
\]
Rewrite this as $|vx|^2 = (|ox|^2 + |o'x|^2 - 2 \, \lfs(v)^2) / 2$.
As $x \not\in B'$, $|o'x|^2 \geq \lfs(v)^2$.
Combining these with the inequality
$|ox|^2 \geq \lfs(v)^2 + 2 \, \lfs(v) \sqrt{\lfs(v)^2 - |vx|^2}$ gives
$|vx|^2 \geq \lfs(v) \sqrt{\lfs(v)^2 - |vx|^2}$.

As $|vx| < \xi \, \lfs(v)$,
we have $\xi^2 > |vx|^2 / \lfs(v)^2 \geq \sqrt{1 - |vx|^2 / \lfs(v)^2}
> \sqrt{1 - \xi^2}$, which is equivalent to
$\xi^4 + \xi^2 - 1 > 0$, hence $\xi > \sqrt{(\sqrt{5} - 1) / 2}$.
The result follows by contradiction.
\end{proof}

The next lemma generalizes Lemma~\ref{lem:raybisector} to higher dimensions.

\begin{lemma}
\label{lem:raybisectorhi}
Consider two distinct sites $v, w \in V$ and a point
$x \in \Vor|_\Sigma v \cap \Vor|_\Sigma w$.
Suppose that $|vx| < \xi \, \lfs(v)$
where $\xi = \sqrt{(\sqrt{5} - 1) / 2} \doteq 0.786151$.
Let $\Lambda$ be the hyperplane that orthogonally bisects the line segment $vw$
(and thus $x \in \Lambda$).
By Lemma~\ref{lem:abovebelowhi},
$T_x\Sigma$ intersects $\mathbb{D}$ at a lone point~$t$.
Let $a$ be the open ray $\vec{xt}$, and observe that $a \subset T_x\Sigma$.

Then $v$ and $a$ are strictly on the same side of~$\Lambda$.
\end{lemma}

\begin{proof}
As $\mathbb{B}$ is disjoint from $\Sigma$, $\mathbb{B}$ contains no sites.
Hence, for every point $o \in \partial \mathbb{D}$, $|vo| \leq |wo|$.
Therefore, every point on $\partial \mathbb{D}$ lies
either on $\Lambda$ or on the same side of $\Lambda$ as~$v$.
(More broadly, $\partial \mathbb{D} \subset \Vor \, v$.)
As $\mathbb{D}$ is an open disk, $v \in \mathbb{D}$, and $v \not\in \Lambda$,
$\Lambda$ does not intersect $\mathbb{D}$.
(More broadly, $\mathbb{D}$ is a subset of the interior of $\Vor \, v$.)
By contrast, $a$~does intersect $\mathbb{D}$ (at $t$).
Recall that $a$'s~origin $x$ lies on $\Lambda$.
Therefore, the open ray $a$ is strictly on the same side of $\Lambda$ as
$\mathbb{D}$, which contains $v$.
\end{proof}

The next lemma generalizes Lemma~\ref{lem:neighborinject}.

\begin{lemma}
\label{lem:neighborinjecthi}
Let $v, x \in \Sigma$ be two points such that $|vx| < \xi \, \lfs(v)$.
Let $\varphi$ be the function that orthogonally projects points of $\R^d$ onto
$v$'s tangent space~$T_v\Sigma$.
There exists an open neighborhood
$N \subset \Sigma$ of $x$ such that $\varphi|_N$ is
a homeomorphism from $N$ to its image $\varphi(N)$ on~$T_v\Sigma$.
\end{lemma}

\begin{proof}
As $|vx| < \xi \, \lfs(v)$, by Lemma~\ref{lem:abovebelowhi},
% (or Lemma~\ref{ptnormal2}),
% $\angle (N_v\Sigma, N_x\Sigma) \neq 90^\circ$.
$T_x\Sigma \cap N_v\Sigma$ contains exactly one point, hence
no line on $T_x\Sigma$ is parallel to $N_v\Sigma$, hence
no line on $T_x\Sigma$ is orthogonal to $T_v\Sigma$.
It follows from the smoothness of $\Sigma$ that
if we choose $N$ to be sufficiently small, $\varphi|_N$ is injective.
As $\varphi|_N$ is injective and $\varphi|_N$ and its inverse are continuous,
$\varphi|_N$ is a homeomorphism.
\end{proof}

The main theorem of this section is almost the same as
Theorem~\ref{thm:homeocell}, but
it applies to a restricted Voronoi cell on
a smooth, closed $k$-manifold $\Sigma$ embedded in the Euclidean space~$\R^d$.
The proof requires no change except to replace the word ``disk'' with
``$k$-ball.''

\begin{theorem}
\label{thm:homeocellhi}
Let $\Sigma \subset \R^d$ be a smooth, closed $k$-manifold.
Consider a site $v \in V$ and its restricted Voronoi cell $C = \Vor|_\Sigma v$.
Suppose that for every point $y \in \Vor|_\Sigma v$,
$|vy| < \xi \, \lfs(v)$, where
$\xi = \sqrt{(\sqrt{5} - 1) / 2} \doteq 0.786151$.

Then $\chi \circ \varphi|_C$~is a homeomorphism
from $\Vor|_\Sigma v$ to a closed unit $k$-ball on~$T_v\Sigma$.
\end{theorem}

A final corollary addressing $\epsilon$-samples,
generalizing Corollary~\ref{cor:homeocelleps} to higher dimensions, follows from
% the Feature Translation Lemma (Lemma~\ref{lem:ftl}),
Theorem~\ref{thm:homeocellhi} and
Lemmas~\ref{lem:pathsincell}, \ref{lem:sixsites}, \ref{lem:raybisectorhi},
and~\ref{lem:ftl} (the Feature Translation Lemma).

\begin{corollary}
\label{cor:homeocellepshi}
Let $\Sigma \subset \R^d$ be a smooth, closed $k$-manifold.
Let $V$ be an $\epsilon$-sample of $\Sigma$ for
$\epsilon < \frac{\xi}{\xi + 1} \doteq 0.440137$.

Then every restricted Voronoi cell in $\Vor|_\Sigma V$ is
homeomorphic to a closed $k$-ball,
the orthogonal projection of $\Vor|_\Sigma v$ onto $T_v\Sigma$ is
star-shaped for all $v \in V$,
$\Sigma$~does not intersect any $(d - 1)$-face of $\Vor \, V$ tangentially,
every connected component of $\Sigma$ has
at least six sites % and six restricted Voronoi cells
on it, and
no restricted Voronoi cell intersects
%%%UNCUT
the interior of another. % restricted Voronoi cell.
\end{corollary}

% \section{Conclusions}

% The nearest-point map doesn't give constants this good.

% \subparagraph*{Acknowledgments.}

% This work was set in motion by a sabbatical at
% INRIA Sophia-Antipolis during the spring of 2010.
% I would like to thank the Geometrica Group for their kind reception, and
% I particularly thank Jean-Daniel Boissonnat and Arijit Ghosh for
% extended discussions of the problem of manifold reconstruction.

\bibliography{rdthomeo}

\end{document}

%% file: vordel2d3d.pdf_t
\begin{picture}(0,0)%
\includegraphics{vordel2d3d.pdf}%
\end{picture}%
\setlength{\unitlength}{1855sp}%
\begin{picture}(13928,19849)(-1139,-19419)
\put(9601,-9811){\makebox(0,0)[b]{\smash{\fontsize{12}{14.4}\usefont{T1}{ptm}{m}{n}{\color[rgb]{0,0,0}restricted Voronoi cells}%
}}}
\put(-1124,-361){\makebox(0,0)[lb]{\smash{\fontsize{12}{14.4}\usefont{T1}{ptm}{m}{n}{\color[rgb]{0,0,0}(a)}%
}}}
\put(6076,-10861){\makebox(0,0)[lb]{\smash{\fontsize{12}{14.4}\usefont{T1}{ptm}{m}{n}{\color[rgb]{0,0,0}(c)}%
}}}
\put(-1124,-9661){\makebox(0,0)[lb]{\smash{\fontsize{12}{14.4}\usefont{T1}{ptm}{m}{n}{\color[rgb]{0,0,0}(d)}%
}}}
\put(-1124,-15361){\makebox(0,0)[lb]{\smash{\fontsize{12}{14.4}\usefont{T1}{ptm}{m}{n}{\color[rgb]{0,0,0}(e)}%
}}}
\put(6076,-361){\makebox(0,0)[lb]{\smash{\fontsize{12}{14.4}\usefont{T1}{ptm}{m}{n}{\color[rgb]{0,0,0}(b)}%
}}}
\end{picture}%

%% file: medial3.pdf_t
\begin{picture}(0,0)%
\includegraphics{medial3.pdf}%
\end{picture}%
\setlength{\unitlength}{1224sp}%
\begin{picture}(20519,8158)(-2657,-9830)
\put(601,-8311){\makebox(0,0)[lb]{\smash{\fontsize{11}{13.2}\usefont{T1}{ptm}{m}{it}{\color[rgb]{0,0,0}$M$}%
}}}
\put(1876,-6811){\makebox(0,0)[lb]{\smash{\fontsize{11}{13.2}\usefont{T1}{ptm}{m}{it}{\color[rgb]{0,0,0}$\Sigma$}%
}}}
\put(1351,-4561){\makebox(0,0)[lb]{\smash{\fontsize{11}{13.2}\usefont{T1}{ptm}{m}{it}{\color[rgb]{0,0,0}$M$}%
}}}
\put(14776,-7036){\makebox(0,0)[lb]{\smash{\fontsize{11}{13.2}\usefont{T1}{ptm}{m}{it}{\color[rgb]{0,0,0}$x$}%
}}}
\end{picture}%

%% file: oopposite2.pdf_t
\begin{picture}(0,0)%
\includegraphics{oopposite2.pdf}%
\end{picture}%
\setlength{\unitlength}{1973sp}%
\begin{picture}(9123,4141)(2011,-3665)
\put(10876,-1036){\makebox(0,0)[lb]{\smash{\fontsize{15}{18}\usefont{T1}{ptm}{m}{it}{\color[rgb]{0,0,0}$w$}%
}}}
\put(2026,-3061){\makebox(0,0)[rb]{\smash{\fontsize{15}{18}\usefont{T1}{ptm}{m}{it}{\color[rgb]{0,0,0}$B'$}%
}}}
\put(4201,-2386){\makebox(0,0)[b]{\smash{\fontsize{15}{18}\usefont{T1}{ptm}{m}{it}{\color[rgb]{0,0,0}$m$}%
}}}
\put(3526,-3511){\makebox(0,0)[rb]{\smash{\fontsize{15}{18}\usefont{T1}{ptm}{m}{it}{\color[rgb]{0,0,0}$o'$}%
}}}
\put(4201,-2761){\makebox(0,0)[rb]{\smash{\fontsize{15}{18}\usefont{T1}{ptm}{m}{it}{\color[rgb]{0,0,0}$B_m$}%
}}}
\put(5401,-2611){\makebox(0,0)[lb]{\smash{\fontsize{15}{18}\usefont{T1}{ptm}{m}{it}{\color[rgb]{0,0,0}$\Sigma$}%
}}}
\put(2026,-1036){\makebox(0,0)[rb]{\smash{\fontsize{15}{18}\usefont{T1}{ptm}{m}{it}{\color[rgb]{0,0,0}$\Sigma$}%
}}}
\put(3526,-1411){\makebox(0,0)[rb]{\smash{\fontsize{15}{18}\usefont{T1}{ptm}{m}{it}{\color[rgb]{0,0,0}$v$}%
}}}
\put(2851, 89){\makebox(0,0)[b]{\smash{\fontsize{15}{18}\usefont{T1}{ptm}{m}{it}{\color[rgb]{0,0,0}$\lfs(v)$}%
}}}
\put(4801,-1861){\makebox(0,0)[lb]{\smash{\fontsize{15}{18}\usefont{T1}{ptm}{m}{it}{\color[rgb]{0,0,0}$x$}%
}}}
\put(5026,-3436){\makebox(0,0)[rb]{\smash{\fontsize{15}{18}\usefont{T1}{ptm}{m}{it}{\color[rgb]{0,0,0}$T_x\Sigma$}%
}}}
\put(2026,-1711){\makebox(0,0)[rb]{\smash{\fontsize{15}{18}\usefont{T1}{ptm}{m}{it}{\color[rgb]{0,0,0}$Q$}%
}}}
\put(7726,-436){\makebox(0,0)[rb]{\smash{\fontsize{15}{18}\usefont{T1}{ptm}{m}{it}{\color[rgb]{0,0,0}$B$}%
}}}
\put(7726,-3061){\makebox(0,0)[rb]{\smash{\fontsize{15}{18}\usefont{T1}{ptm}{m}{it}{\color[rgb]{0,0,0}$B'$}%
}}}
\put(9226,-3511){\makebox(0,0)[rb]{\smash{\fontsize{15}{18}\usefont{T1}{ptm}{m}{it}{\color[rgb]{0,0,0}$o'$}%
}}}
\put(11101,-2611){\makebox(0,0)[lb]{\smash{\fontsize{15}{18}\usefont{T1}{ptm}{m}{it}{\color[rgb]{0,0,0}$\Sigma$}%
}}}
\put(11101,-3211){\makebox(0,0)[lb]{\smash{\fontsize{15}{18}\usefont{T1}{ptm}{m}{it}{\color[rgb]{0,0,0}$\Lambda$}%
}}}
\put(7726,-1036){\makebox(0,0)[rb]{\smash{\fontsize{15}{18}\usefont{T1}{ptm}{m}{it}{\color[rgb]{0,0,0}$\Sigma$}%
}}}
\put(10501,-1861){\makebox(0,0)[lb]{\smash{\fontsize{15}{18}\usefont{T1}{ptm}{m}{it}{\color[rgb]{0,0,0}$x$}%
}}}
\put(9226,-436){\makebox(0,0)[rb]{\smash{\fontsize{15}{18}\usefont{T1}{ptm}{m}{it}{\color[rgb]{0,0,0}$o$}%
}}}
\put(3526,-436){\makebox(0,0)[rb]{\smash{\fontsize{15}{18}\usefont{T1}{ptm}{m}{it}{\color[rgb]{0,0,0}$o$}%
}}}
\put(8776,-736){\makebox(0,0)[rb]{\smash{\fontsize{15}{18}\usefont{T1}{ptm}{m}{it}{\color[rgb]{0,0,0}$a$}%
}}}
\put(9376,-1936){\makebox(0,0)[lb]{\smash{\fontsize{15}{18}\usefont{T1}{ptm}{m}{it}{\color[rgb]{0,0,0}$v$}%
}}}
\put(3526,-2461){\makebox(0,0)[rb]{\smash{\fontsize{15}{18}\usefont{T1}{ptm}{m}{it}{\color[rgb]{0,0,0}$N_v\Sigma$}%
}}}
\put(9226,-2461){\makebox(0,0)[rb]{\smash{\fontsize{15}{18}\usefont{T1}{ptm}{m}{it}{\color[rgb]{0,0,0}$N_v\Sigma$}%
}}}
\put(9226,-1336){\makebox(0,0)[rb]{\smash{\fontsize{15}{18}\usefont{T1}{ptm}{m}{it}{\color[rgb]{0,0,0}$t$}%
}}}
\put(2026,-436){\makebox(0,0)[rb]{\smash{\fontsize{15}{18}\usefont{T1}{ptm}{m}{it}{\color[rgb]{0,0,0}$B$}%
}}}
\end{picture}%

%% file: radialpath.pdf_t
\begin{picture}(0,0)%
\includegraphics{radialpath.pdf}%
\end{picture}%
\setlength{\unitlength}{1973sp}%
\begin{picture}(8819,4301)(204,-8540)
\put(2176,-5161){\makebox(0,0)[lb]{\smash{\fontsize{15}{18}\usefont{T1}{ptm}{m}{it}{\color[rgb]{0,0,0}$a$}%
}}}
\put(6526,-6211){\makebox(0,0)[rb]{\smash{\fontsize{15}{18}\usefont{T1}{ptm}{m}{it}{\color[rgb]{0,0,0}$z$}%
}}}
\put(4426,-8386){\makebox(0,0)[b]{\smash{\fontsize{15}{18}\usefont{T1}{ptm}{m}{it}{\color[rgb]{0,0,0}$\varphi(y)$}%
}}}
\put(6826,-8386){\makebox(0,0)[b]{\smash{\fontsize{15}{18}\usefont{T1}{ptm}{m}{it}{\color[rgb]{0,0,0}$\varphi(z)$}%
}}}
\put(8701,-8386){\makebox(0,0)[b]{\smash{\fontsize{15}{18}\usefont{T1}{ptm}{m}{it}{\color[rgb]{0,0,0}$r$}%
}}}
\put(2551,-6661){\makebox(0,0)[lb]{\smash{\fontsize{15}{18}\usefont{T1}{ptm}{m}{it}{\color[rgb]{0,0,0}$x$}%
}}}
\put(2701,-8386){\makebox(0,0)[b]{\smash{\fontsize{15}{18}\usefont{T1}{ptm}{m}{it}{\color[rgb]{0,0,0}$\varphi(x)$}%
}}}
\put(451,-8011){\makebox(0,0)[rb]{\smash{\fontsize{15}{18}\usefont{T1}{ptm}{m}{it}{\color[rgb]{0,0,0}$q = v$}%
}}}
\put(2176,-6961){\makebox(0,0)[rb]{\smash{\fontsize{15}{18}\usefont{T1}{ptm}{m}{it}{\color[rgb]{0,0,0}$\gamma$}%
}}}
\put(4276,-5836){\makebox(0,0)[lb]{\smash{\fontsize{15}{18}\usefont{T1}{ptm}{m}{it}{\color[rgb]{0,0,0}$y$}%
}}}
\put(526,-5761){\makebox(0,0)[rb]{\smash{\fontsize{15}{18}\usefont{T1}{ptm}{m}{it}{\color[rgb]{0,0,0}$N_v\Sigma$}%
}}}
\put(751,-4636){\makebox(0,0)[lb]{\smash{\fontsize{15}{18}\usefont{T1}{ptm}{m}{it}{\color[rgb]{0,0,0}$t$}%
}}}
\put(3901,-5386){\makebox(0,0)[rb]{\smash{\fontsize{15}{18}\usefont{T1}{ptm}{m}{it}{\color[rgb]{0,0,0}$\Lambda$}%
}}}
\put(3751,-6436){\makebox(0,0)[b]{\smash{\fontsize{15}{18}\usefont{T1}{ptm}{m}{it}{\color[rgb]{0,0,0}$P$}%
}}}
\put(7201,-7111){\makebox(0,0)[b]{\smash{\fontsize{15}{18}\usefont{T1}{ptm}{m}{it}{\color[rgb]{0,0,0}$P'$}%
}}}
\end{picture}%

%% file: lcontinuous.pdf_t
\begin{picture}(0,0)%
\includegraphics{lcontinuous.pdf}%
\end{picture}%
\setlength{\unitlength}{1973sp}%
\begin{picture}(7753,3649)(739,-3248)
\put(7126,-886){\makebox(0,0)[rb]{\smash{\fontsize{15}{18}\usefont{T1}{ptm}{m}{it}{\color[rgb]{0,0,0}$a^+$}%
}}}
\put(8326,-436){\makebox(0,0)[lb]{\smash{\fontsize{15}{18}\usefont{T1}{ptm}{m}{it}{\color[rgb]{0,0,0}$N^+$}%
}}}
\put(4876,-1111){\makebox(0,0)[lb]{\smash{\fontsize{15}{18}\usefont{T1}{ptm}{m}{it}{\color[rgb]{0,0,0}$z^-$}%
}}}
\put(4576,-811){\makebox(0,0)[rb]{\smash{\fontsize{15}{18}\usefont{T1}{ptm}{m}{it}{\color[rgb]{0,0,0}$p_2^-$}%
}}}
\put(7201, 14){\makebox(0,0)[lb]{\smash{\fontsize{15}{18}\usefont{T1}{ptm}{m}{it}{\color[rgb]{0,0,0}$p_2^+$}%
}}}
\put(3451,-1861){\makebox(0,0)[lb]{\smash{\fontsize{15}{18}\usefont{T1}{ptm}{m}{it}{\color[rgb]{0,0,0}$N^-$}%
}}}
\put(4651,-1936){\makebox(0,0)[rb]{\smash{\fontsize{15}{18}\usefont{T1}{ptm}{m}{it}{\color[rgb]{0,0,0}$a^-$}%
}}}
\put(2476,-1636){\makebox(0,0)[b]{\smash{\fontsize{15}{18}\usefont{T1}{ptm}{m}{it}{\color[rgb]{0,0,0}$e$}%
}}}
\put(6076,-1111){\makebox(0,0)[b]{\smash{\fontsize{15}{18}\usefont{T1}{ptm}{m}{it}{\color[rgb]{0,0,0}$z$}%
}}}
\put(2101,-1111){\makebox(0,0)[b]{\smash{\fontsize{15}{18}\usefont{T1}{ptm}{m}{it}{\color[rgb]{0,0,0}$v$}%
}}}
\put(2851,-1111){\makebox(0,0)[b]{\smash{\fontsize{15}{18}\usefont{T1}{ptm}{m}{it}{\color[rgb]{0,0,0}$x$}%
}}}
\put(7276,-2011){\makebox(0,0)[lb]{\smash{\fontsize{15}{18}\usefont{T1}{ptm}{m}{it}{\color[rgb]{0,0,0}$p_1^+$}%
}}}
\put(4576,-2536){\makebox(0,0)[lb]{\smash{\fontsize{15}{18}\usefont{T1}{ptm}{m}{it}{\color[rgb]{0,0,0}$p_1^-$}%
}}}
\put(7351,-1411){\makebox(0,0)[lb]{\smash{\fontsize{15}{18}\usefont{T1}{ptm}{m}{it}{\color[rgb]{0,0,0}$z^+$}%
}}}
\put(2401,-2611){\makebox(0,0)[lb]{\smash{\fontsize{15}{18}\usefont{T1}{ptm}{m}{it}{\color[rgb]{0,0,0}$I_v$}%
}}}
\end{picture}%

%% file: voredges.pdf_t
\begin{picture}(0,0)%
\includegraphics{voredges.pdf}%
\end{picture}%
\setlength{\unitlength}{1973sp}%
\begin{picture}(10546,3166)(285,-2544)
\put(751,-1111){\makebox(0,0)[b]{\smash{\fontsize{15}{18}\usefont{T1}{ptm}{m}{it}{\color[rgb]{0,0,0}$C_1$}%
}}}
\put(5101,-361){\makebox(0,0)[b]{\smash{\fontsize{15}{18}\usefont{T1}{ptm}{m}{it}{\color[rgb]{0,0,0}$C_1$}%
}}}
\put(5251,-1111){\makebox(0,0)[b]{\smash{\fontsize{15}{18}\usefont{T1}{ptm}{m}{it}{\color[rgb]{0,0,0}$C_2$}%
}}}
\put(9301,-286){\makebox(0,0)[b]{\smash{\fontsize{15}{18}\usefont{T1}{ptm}{m}{it}{\color[rgb]{0,0,0}$C_1$}%
}}}
\put(9301,-2011){\makebox(0,0)[b]{\smash{\fontsize{15}{18}\usefont{T1}{ptm}{m}{it}{\color[rgb]{0,0,0}$C_2$}%
}}}
\put(1951,-1861){\makebox(0,0)[b]{\smash{\fontsize{15}{18}\usefont{T1}{ptm}{m}{it}{\color[rgb]{0,0,0}$C_2$}%
}}}
\end{picture}%

%% file: trivertex.pdf_t
\begin{picture}(0,0)%
\includegraphics{trivertex.pdf}%
\end{picture}%
\setlength{\unitlength}{1973sp}%
\begin{picture}(5280,4909)(1329,-8483)
\put(5401,-6811){\makebox(0,0)[lb]{\smash{\fontsize{15}{18}\usefont{T1}{ptm}{m}{it}{\color[rgb]{0,0,0}$\tau$}%
}}}
\put(3751,-6586){\makebox(0,0)[b]{\smash{\fontsize{15}{18}\usefont{T1}{ptm}{m}{it}{\color[rgb]{0,0,0}$r$}%
}}}
\put(4051,-4636){\makebox(0,0)[rb]{\smash{\fontsize{15}{18}\usefont{T1}{ptm}{m}{it}{\color[rgb]{0,0,0}$u$}%
}}}
\put(3301,-5311){\makebox(0,0)[rb]{\smash{\fontsize{15}{18}\usefont{T1}{ptm}{m}{it}{\color[rgb]{0,0,0}$s$}%
}}}
\put(4726,-4936){\makebox(0,0)[lb]{\smash{\fontsize{15}{18}\usefont{T1}{ptm}{m}{it}{\color[rgb]{0,0,0}$s$}%
}}}
\put(5026,-7336){\makebox(0,0)[rb]{\smash{\fontsize{15}{18}\usefont{T1}{ptm}{m}{it}{\color[rgb]{0,0,0}$s$}%
}}}
\put(1726,-7111){\makebox(0,0)[rb]{\smash{\fontsize{15}{18}\usefont{T1}{ptm}{m}{it}{\color[rgb]{0,0,0}$v$}%
}}}
\put(4576,-3961){\makebox(0,0)[lb]{\smash{\fontsize{15}{18}\usefont{T1}{ptm}{m}{it}{\color[rgb]{0,0,0}$n_u$}%
}}}
\put(5401,-8161){\makebox(0,0)[lb]{\smash{\fontsize{15}{18}\usefont{T1}{ptm}{m}{it}{\color[rgb]{0,0,0}$v'$}%
}}}
\put(4126,-3961){\makebox(0,0)[rb]{\smash{\fontsize{15}{18}\usefont{T1}{ptm}{m}{it}{\color[rgb]{0,0,0}$n_\tau$}%
}}}
\put(5401,-5536){\makebox(0,0)[lb]{\smash{\fontsize{15}{18}\usefont{T1}{ptm}{m}{it}{\color[rgb]{0,0,0}$v''$}%
}}}
\put(1351,-5761){\makebox(0,0)[b]{\smash{\fontsize{15}{18}\usefont{T1}{ptm}{m}{it}{\color[rgb]{0,0,0}$n_v$}%
}}}
\end{picture}%

%% file: chains.pdf_t
\begin{picture}(0,0)%
\includegraphics{chains.pdf}%
\end{picture}%
\setlength{\unitlength}{1973sp}%
\begin{picture}(5789,4081)(879,-2784)
\put(5101,-1636){\makebox(0,0)[lb]{\smash{\fontsize{15}{18}\usefont{T1}{ptm}{m}{it}{\color[rgb]{0,0,0}$n_v$}%
}}}
\put(5176,-1036){\makebox(0,0)[lb]{\smash{\fontsize{15}{18}\usefont{T1}{ptm}{m}{it}{\color[rgb]{0,0,0}$n_u$}%
}}}
\put(5626,-286){\makebox(0,0)[b]{\smash{\fontsize{15}{18}\usefont{T1}{ptm}{m}{it}{\color[rgb]{0,0,0}$e$}%
}}}
\put(3601,-1936){\makebox(0,0)[rb]{\smash{\fontsize{15}{18}\usefont{T1}{ptm}{m}{it}{\color[rgb]{0,0,0}$\Sigma$}%
}}}
\put(4201,-1411){\makebox(0,0)[lb]{\smash{\fontsize{15}{18}\usefont{T1}{ptm}{m}{it}{\color[rgb]{0,0,0}$v$}%
}}}
\put(4726,464){\makebox(0,0)[lb]{\smash{\fontsize{15}{18}\usefont{T1}{ptm}{m}{it}{\color[rgb]{0,0,0}$n_\tau$}%
}}}
\put(4126,464){\makebox(0,0)[rb]{\smash{\fontsize{15}{18}\usefont{T1}{ptm}{m}{it}{\color[rgb]{0,0,0}$\tau$}%
}}}
\put(3751,-1411){\makebox(0,0)[rb]{\smash{\fontsize{15}{18}\usefont{T1}{ptm}{m}{it}{\color[rgb]{0,0,0}$w$}%
}}}
\put(4501,-586){\makebox(0,0)[b]{\smash{\fontsize{15}{18}\usefont{T1}{ptm}{m}{it}{\color[rgb]{0,0,0}$u$}%
}}}
\put(2176,-1111){\makebox(0,0)[b]{\smash{\fontsize{13}{15.6}\usefont{T1}{ptm}{m}{n}{\color[rgb]{0,0,0}$F$ inside $\Sigma$}%
}}}
\put(5101,-2311){\makebox(0,0)[b]{\smash{\fontsize{13}{15.6}\usefont{T1}{ptm}{m}{n}{\color[rgb]{0,0,0}$F$ outside $\Sigma$}%
}}}
\end{picture}%